\documentclass[sigconf]{acmart}
\usepackage[utf8]{inputenc} 
\usepackage[T1]{fontenc}    
\usepackage{hyperref}       
\usepackage{url}            
\usepackage{booktabs}       
\usepackage{amsfonts}       
\usepackage{nicefrac}       
\usepackage{microtype}      
\usepackage{amsmath}
\usepackage{graphicx}
\usepackage{natbib}
\usepackage{xcolor}
\usepackage{subcaption}
\usepackage[ruled,linesnumbered]{algorithm2e}

\newtheorem{theorem}{Theorem}
\SetCommentSty{mycommfont}

\input{generic_aliases}

\newtheorem{proposition}[theorem]{Proposition}
\newtheorem{lemma}[theorem]{Lemma}

\newcommand{\highavg}{\left(\frac{1}{2} + \epsilon\right)}
\newcommand{\lowavg}{\left(\frac{1}{2} - \epsilon\right)}
\newcommand{\datavec}{\Vec{\mathbf{y}}}

\newcommand{\nslots}{\bar{k}}

\newcommand{\lapker}{\texttt{GroupFreeKernel}\xspace} 
\newcommand{\groupaware}{\texttt{GroupAware}\xspace} 
\newcommand{\indiv}{\texttt{Individual}\xspace}
\newcommand{\louvain}{\texttt{Louvain}\xspace}
\newcommand{\greedymodu}{\texttt{GreedyModularity}\xspace}

\newcommand{\polblogs}{PolBlogs}
\newcommand{\emaileu}{Email-EU}
\newcommand{\lastfm}{Lastfm-Asia}
\newcommand{\deezer}{Deezer-Europe}

\AtBeginDocument{%
  \providecommand\BibTeX{{%
    \normalfont B\kern-0.5em{\scshape i\kern-0.25em b}\kern-0.8em\TeX}}}

\copyrightyear{2023}
\acmYear{2023}
\setcopyright{acmlicensed}
\acmConference[FAccT '23]{2023 ACM Conference on Fairness, Accountability, and Transparency}{June 12--15, 2023}{Chicago, IL, USA}
\acmBooktitle{2023 ACM Conference on Fairness, Accountability, and Transparency (FAccT '23), June 12--15, 2023, Chicago, IL, USA}
\acmPrice{15.00}
\acmDOI{10.1145/3593013.3594091}
\acmISBN{979-8-4007-0192-4/23/06}

\usepackage{todonotes}
\usepackage{bm}
\renewcommand{\vec}[1]{\bm{#1}}

\begin{document}

\title{Group fairness without demographics using social networks}

\author{David Liu}
\authornote{Work done while at FAIR, Meta AI.}
\email{liu.davi@northeastern.edu}
\affiliation{%
    \institution{Northeastern University}
    \city{Boston}
    \country{USA}
}
\author{Virginie Do}
\email{virginie.do@dauphine.eu}
\affiliation{%
    \institution{FAIR, Meta AI and LAMSADE, PSL, Université Paris Dauphine} 
    \city{Paris}
    \country{France}
}
\author{Nicolas Usunier}
\email{usunier@meta.com}
\affiliation{%
    \institution{FAIR, Meta AI}
    \city{Paris}
    \country{France}
}
\author{Maximilian Nickel}
\email{maxn@meta.com}
\affiliation{%
    \institution{FAIR, Meta AI}
    \city{New York}
    \country{USA}
}

\begin{abstract}
Group fairness is a popular approach to prevent unfavorable treatment of individuals based on sensitive attributes such as race, gender, and disability. However, the reliance of group fairness on access to discrete group information raises several limitations and concerns, especially with regard to privacy, intersectionality, and unforeseen biases. In this work, we propose a ``group-free" measure of fairness that does not rely on sensitive attributes and, instead, is based on homophily in social networks, i.e., the common property that individuals sharing similar attributes are more likely to be connected. Our measure is group-free as it avoids recovering any form of group memberships and uses only pairwise similarities between individuals to define inequality in outcomes relative to the homophily structure in the network. We theoretically justify our measure by showing it is commensurate with the notion of additive decomposability in the economic inequality literature and also bound the impact of non-sensitive confounding attributes. Furthermore, we apply our measure to develop fair algorithms for classification, maximizing information access, and recommender systems. Our experimental results show that the proposed approach can reduce inequality among protected classes without knowledge of sensitive attribute labels. We conclude with a discussion of the limitations of our approach when applied in real-world settings.
\end{abstract}



\maketitle

\section{Introduction}
Fairness has become a central area of machine learning research in order to reduce discrimination and inequalities of outcomes across protected groups in algorithmic decision-making \citep{barocas-hardt-narayanan}.  
Group fairness in particular has received significant attention, where the goal is to reduce the inequality of outcomes between subgroups of a population, e.g., groups defined by demographic attributes, following concepts from anti-discrimination laws \citep{barocas2016big,feldman2015certifying} or distributive justice \citep{roemer2015equality,equality-of-opportunity}.
However, most approaches to promoting equality of outcomes between groups suffer from two weaknesses: First, they are often only suited for a limited number of coarse, predefined groups and do not protect from outcome inequalities at the intersection of these groups, for structured combinations of groups (i.e., fairness gerrymandering) or for groups that are not included in, for instance, anti-discrimination laws but may need protecting \citep{binns2020apparent,wachter2019right}.
Second, they require access to sensitive attributes, which raises privacy concerns and automatically excludes group fairness in scenarios where such information should not be collected or even is legally not permitted \citep{andrus2022demographic}. Addressing \emph{fairness without demographics}, i.e., without access to sensitive attributes, is considered one of the important open problems of algorithmic fairness in practice \citep{holstein2019improving,veale2017fairer,demographic-risks}.

One line of research to approach fairness without demographics assumes access to correlates with sensitive attributes, e.g., features and labels that might be informative about unobserved group memberships. This includes, for instance, work based on proxy features \citep{gupta2018proxy}, as well as methods such as adversarially reweighted learning (ARL)~\citep{lahoti2020fairness}. Access to correlates that are acceptable to use allows to better focus fairness intervention toward relevant groups compared to early worst-case approaches to fairness without demographics \citep{hashimoto2018fairness}. However, it leaves open the question of what these correlates can be in practice, a question that is particularly difficult since understanding how specific variables correlate to groups of interest often requires some level of access to sensitive attributes and dedicated scientific investigations. 

In this paper, we turn to insights from social science and social network analysis to advance fairness without demographics. In particular, we focus on homophily, i.e., the common property of social networks that individuals sharing similar (sensitive as well as non-sensitive) attributes are more likely to be connected to each other than individuals that are dissimilar \citep{mcpherson2001birds}. However, it would be vacuous to infer latent group memberships based on homophily and naively apply the standard group fairness machinery to these labels. While it might help when sensitive attributes are unknown, it would still encounter all the problems of group fairness ranging from privacy to the inadequacy of discrete groups. Instead, we also propose a novel framework to formalize group fairness in a group-free way, with two main characteristics:  a) it does not require any information about sensitive attributes b) it is based purely on the similarity of individuals relative to their links in a social network and c) it does not use or infer any notion of labels or groups for individuals. While these characteristics do not provide formal privacy guarantees per se as, for instance, in differential privacy, we consider them reasonably privacy respecting as they do not use or infer sensitive attributes at any step.

Similar to  the work of \citet{speicher2018unified} our framework is grounded in the decomposition of inequality indices used in economics \citep{cowell2011measuring} that decompose overall inequality into within- and between-group inequality. We show how to extend this decomposition to non-discrete groups, which allows us to generalize the notion of between-group inequality to our case where we only have access to similarity between individuals. In our case, this similarity captures group structure via homophily and we show how these similarity values can be inferred with standard methods using only topological information and without ever accessing or inferring any kind of group information. 
To summarize, our main contributions are as follows:
\begin{description}
    \item[Social Network Homophily] We propose a novel measure of group fairness that does not depend on group labels and instead uses homophily in social networks to reduce inequality in outcomes.
    \item[Group-Free Group Fairness] Our approach is a measure of inequality that is ``group-free" in that it avoids attempting to define groups entirely and is solely based on the similarities of individuals.
    \item[Theoretical Analysis] We support our measure through theoretical analysis and show that it is commensurate with the notion of additive decomposability in the economic inequality literature. We further characterize our measure's performance in the presence of confounding non-sensitive attributes that also drive homophily.
    \item[Empirical Evaluation] We present proof-of-concept experiments using publicly available data on three different tasks: classification, maximizing information access, and recommendation.
\end{description}

As for any method without direct access to group memberships, our method needs to rely on certain assumptions to work as intended. Central to our approach is the assumption of homophily being informative with respect to the groups for which we want to address fairness gaps. We discuss the limitations of this assumption and failure cases in section~\autoref{sec:discussion}.
\section{Method}
In this section, we introduce our approach to measuring inequality among groups based on pairwise similarities between individuals. We first review known results about homophily from social science and discuss how they can be used to capture similarities of users with respect to unobserved attributes. We then introduce a framework that extends group fairness based on inequality indices to the setting where we only have similarity information about individuals. We call this approach "group-free" group fairness as it does not use any discrete group information to measure group inequality.

\subsection{Communities and homophily in social networks}
\label{sec:homophily}
Most real-world networks are known to exhibit community structure, i.e., high concentrations of edges within certain groups of vertices, and low concentrations between these groups \cite{fortunato2010community,girvan2002community}.
A widely recognized cause for community structure in networks is
\emph{homophily}, which refers to the tendency of individuals to connect at a higher rate with people that share similar characteristics to others.
Homophily is considered one of the basic organizing principles of (social) networks and a robust empirical finding in sociology and social network analysis~\cite{%
  lazarsfeld1954friendship,%
  verbrugge1977structure,%
  mcpherson1987homophily,%
  marsden1988homogeneity,%
  burt1991measuring,%
  mcpherson2001birds,%
  kossinets2009origins%
}.
In social networks, homophily has been observed along many dimensions that are relevant in the context of fairness, i.e., dimensions that correspond to sensitive or protected attributes. This includes, for instance, homophily with regard to race and ethnicity~\cite{%
  marsden1987core,%
  marsden1988homogeneity,%
  ibarra1995race,%
  kalmijn1998intermarriage%
},
sex and gender~\cite{mcpherson1986sex,marsden1987core,mayhew1995sex},
age~\cite{verbrugge1977structure,marsden1987core,burt1991measuring},
religion~\cite{laumann1973bonds,verbrugge1977structure,marsden1988homogeneity},
%
%
education, occupation, and social
class~\cite{%
  verbrugge1977structure,%
  marsden1987core,%
  wright1997class,%
  louch2000personal%
}.
Moreover, homophily has not only been observed along these dimensions but also along a wide range of relationship types, e.g., friendships, marriage, work relations, confiding relations, discussion of topics, as well as online social networks \cite{mcpherson2001birds,thelwall2009homophily,ugander2011anatomy}.
However, many networks exhibit community structure beyond what can be explained by homophily on observed attributes. This can, for instance, be driven by homophily on unobserved attributes, heterophily, self-organization, or structural effects such as preferential attachment and stochastic equivalence~\cite{hoff2007modeling,handcock2007model}.

Importantly, we can formalize such notions of community-formation
again in a similarity-based framework. We will base this framework on latent space models, which are a widely-used and flexible approach for this purpose~\cite{hoff2002latent,hoff2007modeling,handcock2007model}.
In particular, let $i \sim j$ denote a possible link in a network $G = (V,E)$ with $|V|=n,$ and let $\vec{z}_{i} \in \mathbb{R}^{d}$ denote the latent features of node $v_i \in V$.
Furthermore, let the conditional probability of observing graph $G$ given latent features $Z = \{\vec{z}_{i}\}_{i=1}^{n}$ be defined as
\begin{equation}
  P\left(G|Z\right) = \prod_{i \neq j} P(i \sim j \,|\,\text{sim}(\vec{z}_{i}, \vec{z}_{j})) ,
  \label{eq:latent-space-prob}
\end{equation}
where $\text{sim}(\cdot , \cdot)$ denotes a model-specific similarity function between latent feature vectors.
A popular parametrization for Equation~\autoref{eq:latent-space-prob} is \emph{latent distance} models where
$P(i \sim j\,|\,\text{sim}(\vec{z}_{i}, \vec{z}_{j})) \propto -\|\vec{z}_{i} - \vec{z}_{j}\|^{2}$
\cite{hoff2002latent,handcock2007model}. Hence, this model class encodes community-formation processes such as homophily through the similarity of latent features.
Another popular class of network models that fits into this framework is
\emph{stochastic blockmodels} (SBMs)~\cite{holland1983stochastic} which encode community structure through latent group memberships.
In particular, an SBM with $n$ nodes and $k$ communities is defined by a pair of matrices $(Z, B)$, where the membership matrix $Z \in \{0,1\}^{n \times k}$ assigns each node to its unique group (or community); and where $B \in [0,1] ^{k \times k}$ is a symmetric connectivity matrix that specifies the probability that a node in community $c_{i}$ is connected to a node in community $c_{j}$.
Let $(Z, B)$ parameterize an SBM and let $P = ZBZ^{\top}$. An SBM parameterizes
Equation~\autoref{eq:latent-space-prob} then as
$P(i \sim j\, |\, \text{sim}(\vec{z}_{i}, \vec{z}_{j})) = p_{ij} = \vec{z}_{i}^{\top} B \vec{z}_{j}$.
The adjacency matrix $A$ of an SBM is generated from $P$ via
${a_{ij} = \text{Bernoulli}(p_{ij}})$ for all $i \neq j$ and $a_{ij} = 0$ otherwise.
Furthermore, the expected adjacency matrix is given by
$\mathbb{E}[A] = P - \text{diag}(P)$.
Stochastic blockmodels can encode network formation processes such as homophily and stochastic equivalence and are one of the most widely used models to study community structure in networks.

\subsection{Group-free group fairness}\label{sec:group-free}
\subsubsection{Background} 
Group fairness aims to minimize the inequality of outcomes among population groups. The traditional setting assumes that these groups are defined by sensitive attributes and form a partition of the population. 

In particular, let us consider a population of size $n$, which we identify with $\{1, \dots,n\}$ and let $x\in\Re^n$ denote the outcome of interest from algorithmic decisions. Let $\mathcal{G}$ be a partition of the population into groups, i.e., $\{1, \ldots,n\} = \cup_{g\in \mathcal{G}} g$ and $g\cap g'=\emptyset$ for every $g\neq g'$ in $\mathcal{G}$. Let $\Delta_b(x)\in\Re_+$ denote a measure of the inequality between the groups in $\mathcal{G}$. Given some parameters $\theta$, a loss function $L$ such as accuracy on a labeled dataset, denoting $x_\theta$ the outcomes induced by $\theta$, the goal of group fairness in prediction problems
can then be formalized as~\citep{speicher2018unified}
\[
\min_\theta \Delta_b(x_\theta) \quad \text{s.t.}\ L(\theta) \leq \delta
\]

These approaches rely on resolving two questions. First, how do we measure inequality in general, and second how do we measure inequalities \emph{between groups} of individuals? Following \citet{speicher2018unified} we follow the extensive literature on measuring inequality in economics \citep{cowell2011measuring} as the starting point of our approach. 

Following axiomatic approaches to design inequality measures \citep{allison1978measures,shorrocks1980class,shorrocks1984inequality,cowell1981inequality}, a function $F:\bigcup_{n\geq 1}\Re_{++}^n\rightarrow \Re_+$ should satisfy certain axioms to be an inequality measure. Some of the most salient are \citep{allison1978measures}: 1) $F$ is symmetric; i.e., invariant by a permutation of the inputs. 2) $F$ is non-negative and $0$ only if $x$ is constant. 3) $F$ is invariant by scale: for every $\alpha>0$, $F(\alpha x) = F(x)$. 4) $F$ is replication invariant, i.e., $F$ stays constant if you repeat the inputs.\footnote{That is, for every integers $n$ and $r>0$, and every $x\in\Re_{++}^n$, $F_{nr}(x^{\bigotimes r}) = F(x)$ where $x^{\bigotimes r} = (x, \dots , x)$ where $x$ is repeated $r$ times. }

While the axioms above clarify basic conditions, the fundamental property of inequality measures is 5) each $F$ satisfies the principle of transfers \citep{dalton1920measurement}, informally meaning that inequality is reduced by taking from a well-off individual and giving it to a worse-off individual.\footnote{Formally, this is equivalent to say that each $F$ is Schur-convex: let $x_{(i)}$ be the $i$-th largest entry of the vector $\Vec{x}$. Then, for two vectors $x, y \in \Re^n$ if $\sum_{i=1}^d y_{(i)} \geq \sum_{i=1}^d x_{(i)} \forall d\in\{1, 2, \cdots, n\}$, then $F(x) \leq F(y)$.}
The final axiom is about additive decomposition \citep{shorrocks1980class} which allows us to decompose the total inequality in a population into within-group inequality $(\Delta_w)$ and between-group inequality $(\Delta_b)$. Formally, let $x\in\Re_{++}^n$, and $\mathcal{G}$ a partition of $\{1,\ldots,n\}$. Denote by
\begin{align}
\begin{split}
\forall g\in \mathcal{G}, x_g=(x_i)_{i\in g} \quad n_g = \card{g} \quad \mu_g =\frac{1}{n_g} \sum_{i\in g} x_i\\
\mu=\frac{1}{n} \sum_{i=1}^n x_i \quad \bar{x}\in\Re_{++}^n \text{ such that } \bar{x}_i = \mu_{g_i}
\end{split}
\end{align}
where $g_i$ above is the group label of individual $i$. Then $F$ is \emph{additively decomposable} if there exists a function $q:\Re_{++}^2\rightarrow\Re_{++}$ such that, for every partition $\mathcal{G}$, we have:
\begin{align}\label{ref:decomposability}
\begin{split}
F(x) &= \Delta_w(x) +\Delta_b(x)\\
\text{where } \Delta_w(x)&=\sum_{g\in \mathcal{G}}\frac{q(\mu_g) n_g}{q(\mu) n}F(x_g) \text{  and } \Delta_b(x)=F(\bar{x}).
\end{split}
\end{align}
The total inequality can thus be decomposed into within-group inequality ($\Delta_w$) and between-group inequality ($\Delta_b$).
Well-known examples of additively decomposable measures are generalized entropy indices \citep{shorrocks1980class} defined as 
\begin{align}
F(x) = \frac{1}{n\alpha(\alpha-1)}\sum_{i=1}^n\left[\left(\frac{x_i}{\mu}\right)^\alpha-1\right] \text{ where $\alpha\not \in \{0, 1\}$},
\end{align}
with limiting equations for $\alpha=0$ and $\alpha=1$. For $\alpha=2$, $F$ is half the squared coefficient of variation, i.e., $F(x)=\frac{\mathrm{var}(x)}{2\mu^2}$, in which case it can be easily verified by the law of total variance that $q(\mu)=\mu^2$ in \eqref{ref:decomposability}. 
Additively decomposable measures define between-group inequality as the contribution of total inequality attributed to the difference of average outcome between groups and is thus a natural basis for group fairness approaches. However, existing works require groups to partition the data, i.e., see the dependence of $\Delta_w$ and $\Delta_b$ on $\mathcal{G}$. 
In the remainder of this section, we describe the novelty of our group-free group fairness approach, in which we define a notion of group fairness by extending the notion of additive decomposability from partitions of the population to similarity kernels. Note that additive decomposability and generalized entropy indices were used previously in the context of algorithmic fairness by \citet{speicher2018unified}, but our goal is different from theirs. Their motivation was to step away from pure group fairness to also account for within-group inequalities. In our work, additive decomposability is used to ground our extension of between-group inequalities from partitions to similarity kernels in the decomposition of total inequality.

\subsubsection{Weighted inequality measures} Notice that given a generalized entropy $F$ and groups $\mathcal{G}$, the between group inequality is equal to: 
\begin{align}
\begin{split}
\Delta_b(x) = F(\bar{x}) &= \frac{1}{n \alpha(\alpha-1)}\sum_{i=1}^n\big((\bar{x}_i/\mu)^\alpha-1)\\
&= \sum_{g\in \mathcal{G}} \frac{n_g}{n\alpha(\alpha-1)} \big((\mu_g/\mu)^\alpha-1)    
\end{split}
\end{align}
It is then useful to consider an extension of $F$ that admits a weight vector $w\in \Re_{++}^n$: 
\begin{align}\label{eq:weighted_entropy}
F(x, w) = \frac{1}{\alpha(\alpha-1)}\sum_{i=1}^n\frac{w_i}{||w||_1}\Big(\big(\frac{x_i}{\mu}\big)^\alpha-1\Big)
\end{align}
so that denoting $\bar{\mu} = (\mu_g)_{g\in \mathcal{G}}$ and $\bar{w} =(n_g/n)_{g\in \mathcal{G}}$, we have $\Delta_b(x) = F(\bar{\mu}, \bar{w})$. 

Note that this weighted version makes it clear that groups are weighted by their size. This group weighting is inherent to the decomposition approach to defining between-group inequality, since each \emph{individual} has the same weight in the total inequality $F(x)$, as also discussed by \citet{speicher2018unified}. Giving proportionally more weight to small groups could also be reasonable in some contexts. It is compatible with the decomposition approach by reweighting \emph{individuals} in the initial inequality, effectively giving more weight to individuals from smaller groups. We do not explore this further in our experiments.

\subsubsection{Group-free group fairness}\label{sec:kernel}
The main challenge of defining a group-free inequality metric is to replace partition ($\in\mathcal{G}$) and averaging ($\bar{x}$) functions as used in Equation \autoref{ref:decomposability} without relying on group information.
Intuitively, we will use pairwise similarities to replace averages that are partitioned by groups with averages over the entire population but weighted by pairwise similarities. Hence, for an individual $i$, individuals that are similar to $i$ (i.e., closely connected in the network and likely members of similar groups) are weighted stronger in the average than dissimilar individuals. We recover the discrete group setting in the special case where individuals in a group are considered completely similar (weight = 1) while individuals of different groups are completely dissimilar (weight = 0), and in this case, our approach yields the standard between-group inequality. Note that the individual fairness framework of \citet{dwork2012fairness} also uses a similarity between individuals to reduce the inequality of outcomes. However, apart from the usage of similarity metrics, our approaches are entirely different. In individual fairness, similarity is supposed to capture merit and similar individuals should receive similar outcomes irrespective of their group. In contrast, in our work similarity is supposed to reflect group structure and similar outcomes should be distributed across these groups, i.e., across \emph{dissimilar} individuals.

Formally, we define a kernel from a similarity matrix as follows. Let $S \in \Re_+^{n \times n}$ be the pairwise similarity matrix such that $S_{ij} = \text{sim}(i, j)$. 
Furthermore, let $K \in [0,1]^{n \times n}$ be a column-normalized kernel of $S$, i.e., 
$    K = S D^{-1}$
where $D$ is a diagonal matrix with entries $d_{jj} = \sum_i \text{sim}(i, j)$. 
Column normalization is necessary for preserving additive decomposability since it corresponds to the total weight of an individual within the averaging performed by the kernel. Given a kernel $K\in[0,1]^{n\times n}$ and an outcomes vector $\datavec\in\Re^n$, the averaging operator $A$ is defined as
\begin{equation}
   A\left(K, \datavec\right) = \frac{K \datavec}{K\Vec{1}}
    \label{eq:weighted-neigh}
\end{equation}
In the above, the division is applied element-wise and $\Vec{1}\in\Re^n$ is the all-1 vector. Next, we define our measure of group-free group fairness over these smoothed values based on a weighted generalized entropy \eqref{eq:weighted_entropy} by: 
\begin{equation}\label{eqn:group-free-def}
    \Delta_b(\datavec) = F\left(A(K, \datavec), K\Vec{1}\right)
\end{equation}

Our measure of group-free group fairness generalizes the discrete setting: When sensitive attributes are available and define a partition into groups $g \in \mathcal{G}$, we define the \emph{``ground-truth" kernel} as the kernel $K^*$ where $K^*_{ij} = \frac{1}{\lvert g \lvert}$ if $i, j\in g$, and $K^*_{ij} = 0$ otherwise. We also define the \emph{``ground-truth" inequality} as $\Delta^0_b$ calculated with the ground-truth kernel. When using such a kernel, it is immediate to check that the definitions of $\Delta_b$ in Equation~\autoref{eqn:group-free-def} and Equation~\autoref{ref:decomposability} match.

More generally, the definition of \eqref{eqn:group-free-def} stems from the following generalization of additive decomposability, the proof of which can be found in Appendix \ref{sec:decomposability}:
\begin{proposition}[Generalized additive decomposability]\label{prop:additive-decomposability}
    For any generalized entropy $F$, using $q$ defined in Equation \autoref{ref:decomposability}, for every column-normalized $K\in\Re_+^{n\times n}$, i.e., such that $\forall j, \sum_{i=1}^n K_{ij}=1$, for every $\datavec\in\Re_{++}^n$, we have
    \begin{equation}\label{eqn:prop-decomposability}
        F(\datavec) = \left[\sum_{i=1}^n \frac{q\big(A(K, \datavec)_i\big) K_i^\top\Vec{1}}{q(\mu) n} F\left(\datavec, K_i\right)\right] + F\left(A\left(K, \datavec\right), K\Vec{1}\right)
    \end{equation}
\end{proposition}
Note that as stated before, when the kernel corresponds to a partition, then the decomposition in Prop. \ref{prop:additive-decomposability} is the same as the original decomposition of Equation \autoref{ref:decomposability}.

\subsection{Inferring Pairwise Similarities based on Homophily in Networks}\label{sec:kernel-construction}
In subsections \ref{sec:homophily} and \ref{sec:group-free}, we discussed how pairwise similarities based on homophily can be used to define group-free group fairness. In the following, we will discuss how these similarities can be inferred from a given network. 
In particular, given a graph $G = (V, E)$ with $|V|=n$, we are interested in inferring pairwise similarities ${\text{sim} : V \times V \to \mathbb{R}}$ that capture the homophily structure of the network.
To infer $\text{sim}(\cdot, \cdot)$ from $G$, we can use a variety of models ranging from latent distance and eigenmodels as discussed previously~\cite{hoff2002latent,hoff2007modeling}, 
to the stochastic blockmodel via spectral clustering~\citep{newman2006finding,sbm-recovery,amini2013pseudo,zhang2018understanding}, to link prediction models~\citep{hoff2007modeling,grover2016node2vec,nickel2011three} and graph
neural networks~\citep{kipf2016semi,hamilton2017inductive}.
Unless noted otherwise, we will employ Laplacian Eigenmaps, as it is a simple and theoretically well-founded method to infer node similarities for networks with a latent stochastic block-model structure \citep{sbm-recovery}. However, we note explicitly that we do not propose a single best method to infer $\text{sim}(\cdot, \cdot)$ from $G$, as the best approach will be network dependent.

For a graph $G$, let $Z \in \Re^{n \times d}$ be the Laplacian Eigenmap embedding of $G$ and $\vec{z}_i$ be the $i^{\text{th}}$ row of $Z$, such that the columns of $Z$ are the eigenvectors corresponding to the $d$ smallest non-zero eigenvalues of the degree-normalized Laplacian of $G$ \citep{lapeigenmap}. Intuitively, Laplacian Eigenmaps penalizes nodes connected with edge weight $w_{ij}$ for being embedded far apart, i.e.:
\begin{equation}
    Z = \argmin_{Z} \sum_{i,j}^n w_{ij} \|\vec{z}_i - \vec{z}_j\|^2 ~~\text{s.t.}~~ Z^TZ = I 
\end{equation}
It is well known that Laplacian Eigenmap embeddings are well-suited to recover homophily-based similarity in the stochastic block model, e.g., they can be used in combination with spectral clustering for exact recovery of the planted community structure under conditions that match
the information-theoretic limits~\citep{sbm-recovery, deng2020strong}.

Given $Z$, we compute $S \in \mathbb{R}^{n \times n}$ via the cosine similarity of the embeddings, i.e., by setting 
\[
S_{ij} = \text{sim}(i,j) = \frac{\langle \vec{z}_i, \vec{z_j} \rangle}{\|\vec{z_i}\| \|\vec{z}_j\|}
\]
Finally, we scale the similarity scores to be in $[0, 1]$ and column-normalize the kernel matrix:
\begin{align}\label{eqn:kernel-def}
        K = n S_u D^{-1} \quad \text{where} \quad
        S_{u} = \frac{S - \min S}{\max S - \min S}
\end{align}

See Figure~\autoref{fig:sbm-kernel} for an example of inferring $K$ from $G$ with Laplacian Eigenmaps for two synthetic SBMs and its behavior for misspecified embedding dimensions.
\begin{figure}
     \centering
     \begin{subfigure}[b]{0.49\textwidth}
         \centering
         \includegraphics[width=\textwidth]{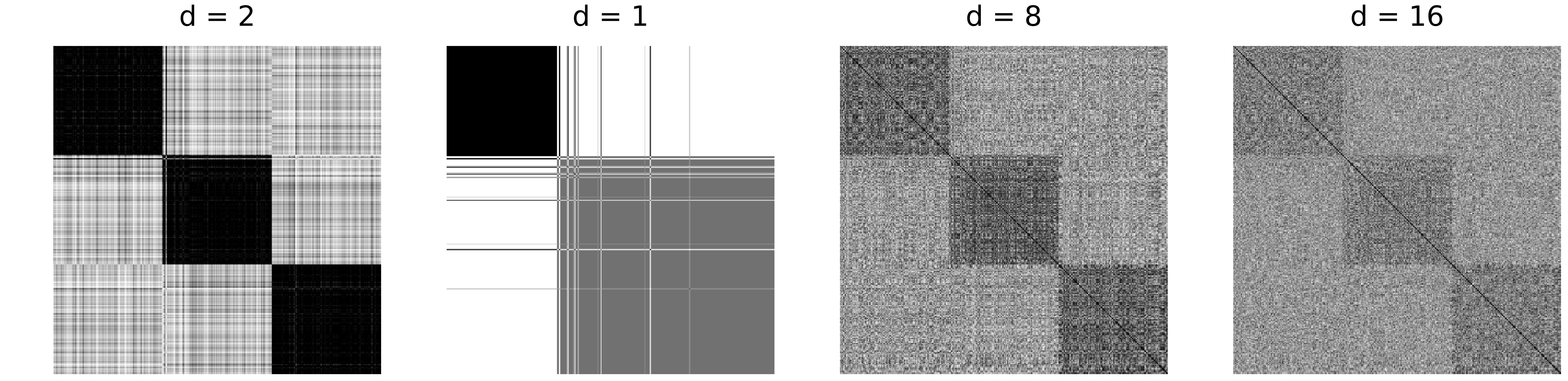}
         \caption{Dimensions exceed groups (groups = 3)}
         \label{fig:kernel_over}
     \end{subfigure}
     \hfill
     \begin{subfigure}[b]{0.49\textwidth}
         \centering
         \includegraphics[width=\textwidth]{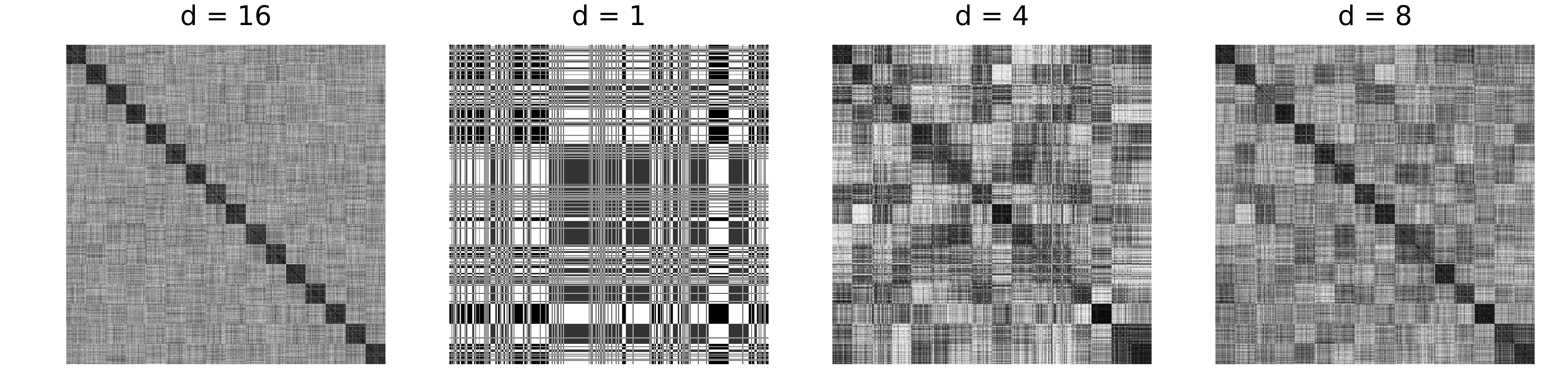}
         \caption{Groups exceed dimensions (groups = 16)}
         \label{fig:kernel_under}
     \end{subfigure}
    \caption[]{
    \textbf{Kernel Inference} Examples of inferred kernels $K$ (Eq. \ref{eqn:kernel-def}) for two SBMs with varying embedding dimensions $d$. When $d$ is specified adequately, the kernel recovers the correct latent block structure. When $d$ is (a) greater than or (b) less than the number of groups, recovery degrades gracefully. In both examples, $G$ is a sample of an SBM with blocks of size $n=100$, intra-group connection probability $p=0.2$, and inter-group connection probability $q=0.05$. Even when $d=16$ is largely overspecified in (a), nodes in the same block share visibly greater similarity. In (b), the block structure is largely recovered with half as many dimensions ($d=8$) as blocks.}
    \label{fig:sbm-kernel}
\end{figure}

It is worth noting that when using networks for the purpose of group fairness, a naive approach could be to simply define discrete groups via community detection and apply existing group fairness measures. 
However, this approach would clearly infer group memberships and possibly violate privacy expectations of individuals. 
In contrast, our approach does not infer any form of group membership not even in intermediate steps. Moreover, community detection can fall short for common network structures that arise in social settings such as core-periphery structures and hierarchical communities. In addition, the assignment to discrete groups is error-prone and sensitive to hyperparameters. In subsection~\autoref{sec:influence-max}, we supplement these conceptual concerns with practical pitfalls we observed when applying community detection for group fairness.

\subsection{Analysis of non-sensitive homophilous attributes}
Our model assumes that every individual in the population possesses latent sensitive attributes, which are sources of homophily. However, it is likely that non-sensitive attributes contribute also to homophily in a network. In these cases, our kernel will recover both sensitive and confounding attributes. In this section, we bound the impact of these confounding attributes under specific assumptions. 

We analyze the specific case in which a single sensitive attribute partitions the population into two classes of equal sizes. Further, the outcome vector $\datavec\in \{0, 1\}^n$ has an equal number of positive and negative labels where $\sum_{i=1}^n \datavec_i = \frac{n}{2}$. Without loss of generality, one of the sensitive classes has $\mu_1 = \frac{n}{2}(1/2 + \epsilon)$ positive labels and another has $\mu_2 = \frac{n}{2}(1/2 - \epsilon)$, where $\epsilon\in [0, 1/2]$ is an inequality parameter. The ground truth inequality is:
\begin{align}\label{eqn:ground-truth-inequality}
\begin{split}
    \Delta_b^0 = F(A(K^*, \datavec), K^*\vec{1})
    =& F\left(\left(\mu_1^{\otimes n/2}, \mu_2^{\otimes n/2}\right), \frac{n}{2}\vec{1}\right)\\ 
    &= \left(\frac{\mu_1 - \mu_2}{\mu_1 + \mu_2}\right)^2\\ 
    &= 4\epsilon^2
\end{split}
\end{align}
Next, we consider confounding attributes under the constraint that each confounding attribute partitions the population into groups of equal size. This restriction allows for a more tractable analysis of the column-normalized kernel matrix $K$. With the above setup and restrictions, we bound the impact of confounding attributes in Proposition \ref{prop:inequality-bounds}, which is proved in Appendix \ref{sec:proof-bounds}.

\begin{proposition}[Inequality Bounds] \label{prop:inequality-bounds}
        In the above setup, where the ground-truth inequality is $\Delta^0_b = 4\epsilon^2$, let us parameterize the kernel as the sum of a kernel $K_s$ corresponding to the sensitive attribute and kernels corresponding to a set of confounding attributes $C$:
        \begin{equation}
            K = K_s + \sum_{C}K_c
        \end{equation}
        The kernel $K_s$ is defined such that $K_s(i, j) = p$ if $i, j$ are in the same sensitive group, otherwise $K_s(i,j) = 0$. Similarly, for a confounding kernel $K_c$, $K_c(i,j) > 0$ if $i$ and $j$ share the same value for confounding attribute $c$. We assume the sum of the absolute values of the confounding kernels is bounded by $\lvert \sum_C K_c\lvert_1 \leq qn^2/2$. Further, each confounding attribute partitions the population into groups of equal size and weight i.e. $K_c$ is a node permutation of a block-diagonal matrix and $K_c\Vec{1} = \lambda_c\Vec{1}$ for a constant $\lambda_c$. Then, the group-free inequality value that our method returns, $\Delta'_b$, is bounded as:
        \begin{equation}
            \Delta'_b \in \left[\left(\frac{p}{p + q}\right)^2 \Delta_b^0, \Delta_b^0 + \left(\frac{q}{p + q}\right)^2\left(1 - \Delta_b^0\right)\right]
        \end{equation}
\end{proposition}
The above bounds are functions of $p, q, \epsilon$. To be more concrete, we provide visualizations of the bounds for specific instantiations of the three variables in Figure \ref{fig:inequality-bounds}.
\begin{figure}
    \centering
    \includegraphics[width=\linewidth]{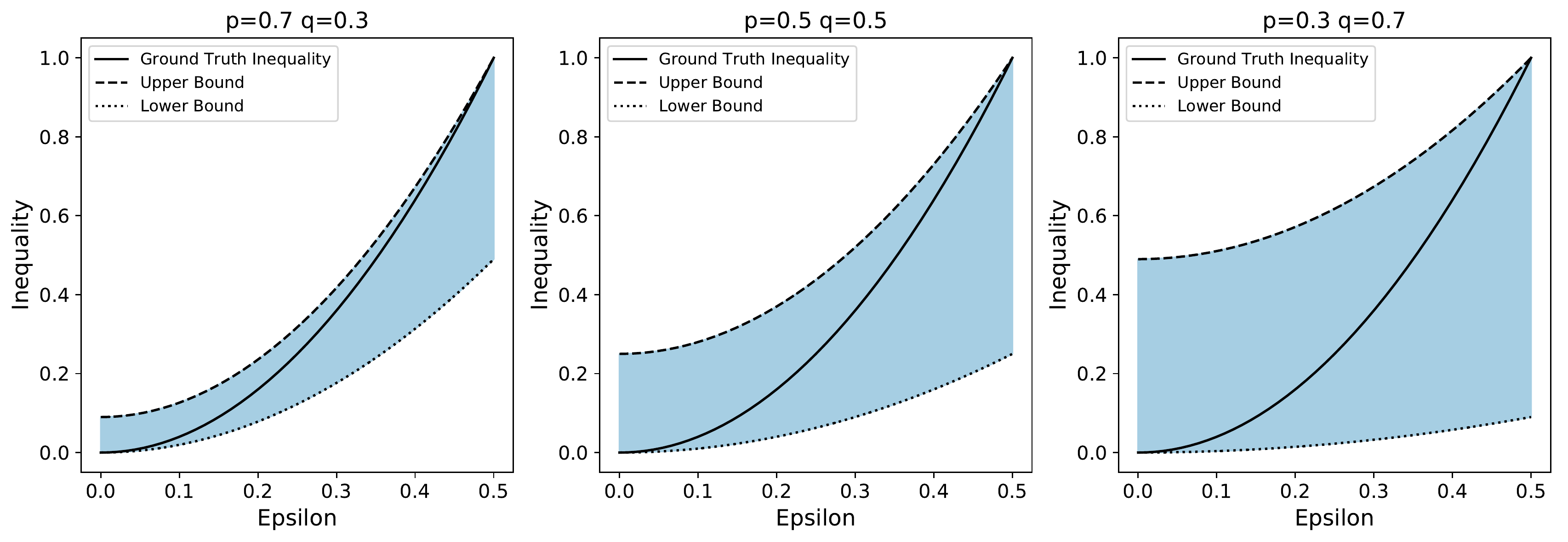}
    \caption[]{\textbf{Inequality Bounds} Instantiations of upper and lower bounds for three pairs of $p, q$. From left to right: sensitive attribute similarities dominate the confounding similarities; similarities are of equal value; and, confounding similarities dominate the sensitive ones. Epsilon (x-axis) denotes the difference in the prevalence of positive labels between the two sensitive-attribute groups. ``Ground Truth Inequality" denotes group inequality as measured using the sensitive attribute, and our similarity-based measure returns a value in the shaded region. As the strength of the confounding attributes increases, the bounds become wider.}
    \label{fig:inequality-bounds}
\end{figure}
\section{Applications}
Our proposed group-free group fairness measure is adaptable for a variety of tasks as an objective or a constraint. To demonstrate its versatility, we apply the measure to three diverse tasks: node classification, maximizing information access, and recommender systems. In all experiments, we leverage datasets that provide both a social network and individual sensitive attribute labels, henceforth referred to as ``ground-truth" labels. We use the ground-truth labels to show that our kernel-based approach, which utilizes only the network, does indeed lower inequality among these classes. Further, for the maximizing information access and recommender system tasks, we compare our results against the baseline of inferring group memberships via community detection instead of our group-free approach.

\subsection{Data}
Table \ref{tab:datasets} lists the four datasets used in our experiments. Each provides a social network and node-level ground-truth labels. The \polblogs{} network consists of political blogs active during the 2004 U.S. presidential election where the label is the political affiliation of the blog (liberal or conservative), and an edge connects two blogs if one blog includes a hyperlink to the other \cite{polblogs}. \emaileu{} is an email network among members of a large European research institution where edges connect the sender and recipient of an email, and members are labeled according to their department within the research institution \cite{email}. The \lastfm{} network includes data from the Last.fm music service from users in Asian countries \cite{feather}. It includes a social network of users where edges are mutual follows, as well as the list of artists listened to by the users and the country of the users. We also use the \deezer{} dataset \citep{feather} which includes similar preference and network data as Lastfm-Asia, but from the Deezer music streaming service.

Preprocessing details are available in Appendices~\autoref{sec:app-exps} and \ref{app:ranking}. The statistics shown in Table \ref{tab:datasets} are after pre-processing.

\begin{table}[b]
    \small
    \centering 
    \caption[]{\textbf{Network datasets with node-level sensitive attributes}. Statistics for the number of nodes ($\lvert V\lvert$), edges ($\lvert E\lvert$), and groups ($\lvert \mathcal{G}\lvert$) are after pre-processing. Assortativity ($r$) is calculated with respect to the sensitive attribute (higher values suggest greater homophily).}
    \label{tab:datasets}
    \begin{tabular}{llrrrr}
        \toprule
        \bf Dataset & \bf Sensitive \bf Attr. & $\lvert V\lvert$ & $\lvert E\lvert$ & \bf $\lvert \mathcal{G}\lvert$ & \bf $r$\\ \midrule
        \polblogs{} \cite{polblogs}& Political Party & $1,222$ & $19,024$ & $2$ & $0.81$ \\
        \emaileu{} \cite{email} & Department & $339$ & $7,066$ & $8$ & $0.72$ \\
        \lastfm{} \cite{feather} & Country & $2,785$ & $17,017$ & $9$ & $0.90$\\
        \deezer{} \cite{feather} & Gender & $1,090$ & $3,623$ & $2$ & $0.02$\\
        \bottomrule
    \end{tabular}
\end{table}

\subsection{Node Classification}
Let $f$ be a binary classifier that takes as input node features $x_i\in\Re^d$ and outputs a binary label $\mathbf{\hat{Y}}_i \in \{0, 1\}$ for node $i \in G$. Suppose that a positive label is a more desired classification, such as access to a resource of interest. We apply our group-free group fairness measure to ensure that the positive classification labels are distributed throughout $G$. Following prior work in group fairness such as \citep{equality-of-opportunity}, we achieve this via post-processing of the prediction labels. The objective is to produce a re-labeled vector $\mathbf{\Tilde{Y}} \in \{0, 1\}^n$ that modifies as few labels as possible while reducing inequality in the allocation of positive labels. We formulate this as a mixed-integer program over the variable $\mathbf{\Tilde{Y}}$:
\begin{align*}
    &\argmin_{\Tilde{Y} \in \mathbb{R}^n} \sum_{i=1}^n \left[ \left(1 - 2\mathbf{\hat{Y}}_i\right)\mathbf{\Tilde{Y}}_i + \mathbf{\hat{Y}}_i\right]\\
    &\text{s.t.}\quad \mathbf{\Tilde{Y}} \in \{0, 1\} \\ 
    &A(K, \mathbf{\Tilde{Y}}) \geq \theta_{min} \tag{\hbox{soft group minimum exposure}}\\
    &\sum\nolimits_{i=1}^n\mathbf{\Tilde{Y}}(i) \leq \sum\nolimits_{i=1}^n\mathbf{\hat{Y}}(i) \tag{\hbox{preserve num. of positive labels}}
\end{align*}
The objective function penalizes perturbations to the prediction labels, as $ \left(1 - 2\mathbf{\hat{Y}}_i\right)\mathbf{\Tilde{Y}}_i + \mathbf{\hat{Y}}_i = 1$ if $\mathbf{\hat{Y}}_i \ne \mathbf{\Tilde{Y}}_i$. Next, the minimum exposure constraint requires that every weighted average, the fraction of nodes labeled positively, be at least $\theta_{min}$. We do not directly constrain $\Delta_b$ as it is not linear but note that lower bounding the weighted averages indirectly constrains $\Delta_b$. Finally, we preserve the total number of positive labels in the entire network. 
\subsubsection{Results}
We solved the mixed-integer program on the \polblogs{} and \emaileu{} networks. Figure \ref{fig:classification} visualizes the post-processed node labels as the minimum threshold is gradually increased until a feasible solution does not exist. For both networks, we initialize the node labels such that the positive labels are allocated exclusively to select ground-truth groups. As the minimum threshold increases, the labels are progressively distributed throughout the network. For \polblogs{}, the ground-truth inequality ($\Delta_b^0$) decreases from $0.85$ to $0.04$; at the highest threshold, $32\%$ of node labels are modified by post-processing. Similarly, for \emaileu{}, $\Delta_b^0$ decreases from $2.45$ to $0.09$, where at the maximum threshold, $29\%$ of labels are modified.
In Appendix \ref{sec:node-classification-node-features}, we analyze correlations between node features, such as degree and betweenness centrality, and prevalence of positive labels following post-processing.
\begin{figure}
     \centering
     \begin{subfigure}[b]{0.49\textwidth}
         \centering
         \includegraphics[width=\textwidth]{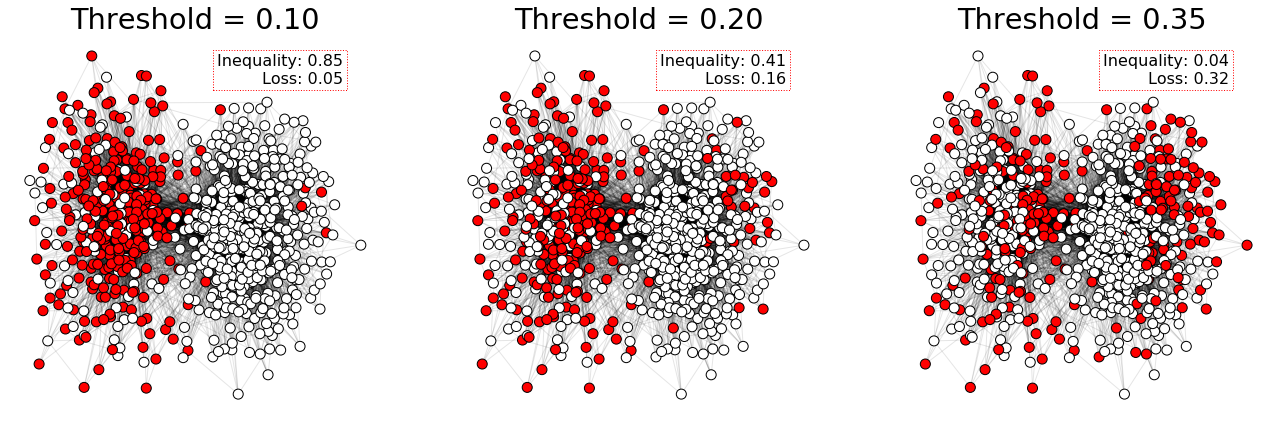}
         \caption[]{\polblogs{}}
         \label{fig:classification-polblogs}
     \end{subfigure}
     \begin{subfigure}[b]{0.49\textwidth}
         \centering
         \includegraphics[width=\textwidth]{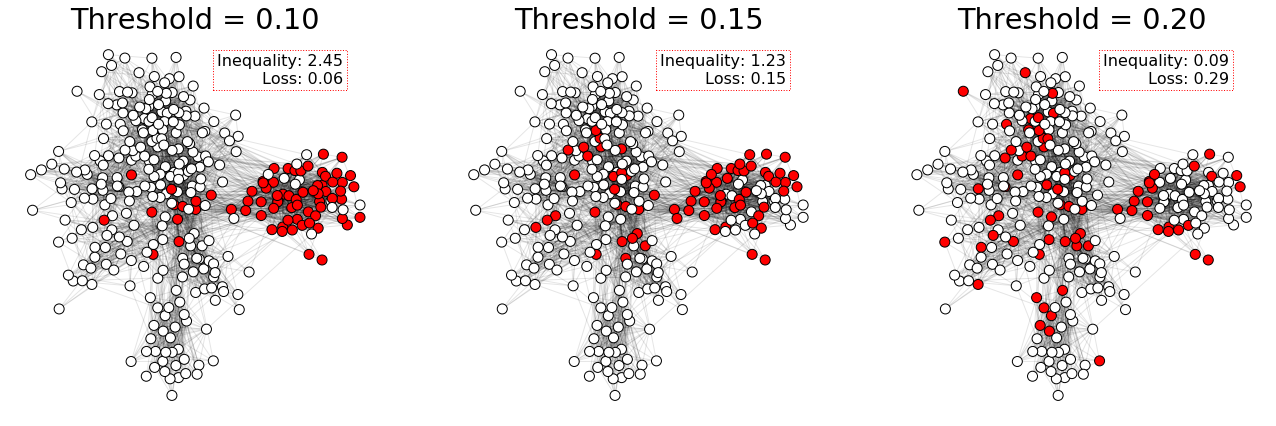}
         \caption[]{\emaileu{}}
         \label{fig:classification-email}
     \end{subfigure}
    \caption[]{
    \textbf{Classification} Post-processing of classification results with our approach on two real-world networks and varying minimum exposure thresholds. Red nodes indicate positive outcomes, becoming increasingly dispersed across groups with higher thresholds. The loss values are the fraction of original labels perturbed in post-processing.
    }
    \label{fig:classification}
\end{figure}

\subsection{Maximizing information access}\label{sec:influence-max}
\subsubsection{Setting}
This task is concerned with selecting "seed" nodes in graph $G$ such that access to the seeded information is maximized across groups after the information diffuses in the network. This objective is especially relevant when important information related to, for instance, health and employment should be distributed fairly in a network.
The information diffusion proceeds is often assumed as follows: at time step $t = 0$ the $k$ seed nodes receive the information and become informed; at future time steps $t$, each node that received the information in step $t-1$ shares the information with each uninformed neighbor with probability $p$; if no nodes receive the information at a time step, the cascade terminates. It is known that selecting the optimal $k$ seed nodes is NP-hard, but, as the objective is submodular, an approximation algorithm that iteratively chooses seeds by greedily maximizing reach -- the expected number of individuals who receive the information -- yields a reach within $1 - 1/e$ of the optimal reach \cite{kleinberg2003influence}.

Several past works have developed algorithms for fair information access in this setting. Many of these algorithms assume access to group labels and aim to balance the information reach among groups so that no group disproportionately receives the information \cite{stoica2020seeding, tsang2019group, rahmattalabi2021welfare}. In contrast \citet{fish} does not assume access to community labels and takes an individual fairness approach. In \citet{fish}, for a given seed set, each node $v$ has a probability $p_v$ of receiving the information; the algorithm iteratively selects the seeds that maximize the minimum value of $p_v$ across all nodes.

\subsubsection{Method}
We present a greedy algorithm for fair information access maximization using our group-free fairness notion. Let $\mathcal{S}$ be the current seed set and let $i$ be a candidate new seed not in $\mathcal{S}$. Further, let $p_v\left(\mathcal{S} \cup s'\right)$ be the probability that node $v$ is activated when an independent cascade starts at the seed set $\mathcal{S} \cup s'$. The smoothed activation probabilities are then $A\left(K,  p\left(\mathcal{S} \cup s' \right)\right)$. Building on approaches that maximize the activation of the worst-off group \cite{tsang2019group}, we iteratively select seeds by maximizing: 
\begin{equation}\label{eqn:greedy-ic}
    \argmax_{s'\in V\setminus \mathcal{S}} \left[ \min A\left(K,  p\left(\mathcal{S} \cup s'\right)\right) \right]
\end{equation}
We repeated the experiment with multiple values of $d$, the number of embedding dimensions, and report results for the setting that most reduces ground-truth inequality; we show in Appendix \ref{sec:influence-max-additional} that alternative values of $d$ yield comparable results. Further, when ground-truth labels are not available to tune $d$, general methods to tune the number of dimensions for the homophily model of choice can be used. In the case of Laplacian Eigenmaps, the number of dimensions can be tuned based on the graph spectra \cite{lee2014multiway, ahmed2013distributed}.

We compare our algorithm against three baselines: vanilla information access maximization (``Greedy Reach"), the individual fairness algorithm from \citet{fish} (``Fish et al."), and community detection (``Community"). The community detection baseline also uses the network but defines groups based on detected communities instead based on node similarities. Specifically, it optimizes Eq. \ref{eqn:greedy-ic} with the community-detection kernel, which is analogous to the ground-truth kernel defined in subsection \ref{sec:kernel} but with groups recovered by the Louvain \cite{blondel2008fast} community detection algorithm.

\subsubsection{Results}
Figure \ref{fig:ic} shows the results of our evaluation.
Across all algorithms and graphs, the ground-truth inequality generally decreases as the seed set increases, as the cascade approaches maximal coverage. However, for a given seed budget, our group-free algorithm is generally able to achieve the lowest ground-truth inequality (also note that plots are on log-log scale). In the case of \polblogs{}, our algorithm achieves a substantial improvement in lowering inequality for all seed budgets and for \lastfm, our algorithm achieves the lowest inequality for seed budgets larger than $10$ nodes. In the case of \emaileu, our algorithm yields lower inequality for most seed budgets larger than $20$ nodes. 

\begin{figure}
    \centering
    \includegraphics[width=\linewidth]{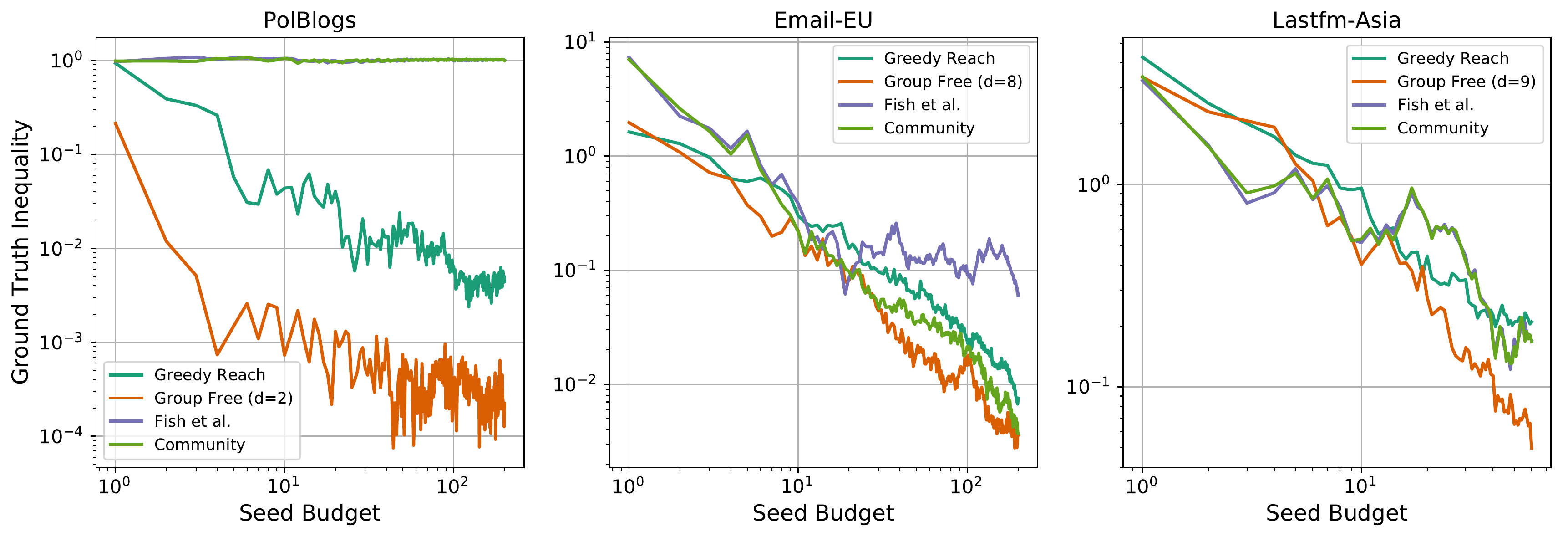}
        \caption[]{\textbf{Maximizing Information Access} Comparison of our approach (``Group-Free") to 
        vanilla information access maximization (``Greedy Reach"), individually fair information access maximization (``Fish et al."), and community detection (``Community") on real-world networks. Plots are log-log scale, showing the seed budget (x-axis) vs inequality among predefined groups (y-axis). Lower values are better.
    }
    \label{fig:ic}
\end{figure}

Not only do the maximizing information access experiments attest to our measure's ability to lower ground-truth inequality, they also surface potential pitfalls of the community detection approach. This is most clearly seen in the case of \polblogs{} for which the community detection baseline produces very high inequality. This behavior arises from small and disconnected communities being erroneously inferred by the detection algorithm.
When using such discrete partitions of the graph as groups, objectives that measure outcomes for the worst-off group become all-or-nothing objectives and may fail since the cascade is unlikely to reach all detected groups (which can be arbitrarily bad). 
In contrast, our group-free measure does not incur such all-or-nothing costs as it is based on continuous similarity values.\footnote{Although not considered in prior work, we extend in Appendix \ref{sec:influence-max-additional} community detection baselines to continuous settings and compare it to our proposed approach. Results are largely in line with the results reported in this section and support the benefits of our approach beyond such discretization issues.}
\subsection{Fairness for users in recommender systems} \label{sec:ranking}

\subsubsection{Setting}

We consider a recommendation task with $m$ items to rank for each of the $n$ users of the social network $G$. The preference score of user $i$ for item $j$, which is typically obtained from a learned value model, is denoted by $\rho_{ij} \in [0,1].$ Given the scores $(\rho_{ij})_{\substack{1\leq i \leq n \\ 1 \leq j \leq m}}$ as input, the recommender systems produce a stochastic ranking policy $P \in [0,1]^{n \times m \times m}$ such that $P_i$ is a bistochastic matrix, and $P_{ijk}$ is the probability to show item $j$ to user $i$ at rank $k$ \citep{singh2018fairness,do2021two}. Formally, the set of stochastic ranking policies is defined as:
$\mathcal{P} = \{P  \in [0,1]^{n \times m \times m}: \, \forall i \in \intint{n}, \,\forall k \in \intint{m},\sum_{j=1}^m P_{ijk} = 1\, \text{ and } \forall j \in \intint{m}, \sum_{k=1}^m {P_{ijk} = 1} \}.$

Following \citep{singh2018fairness,morik2020controlling,do2021two,biega2018equity}, we consider a position-based user model where the probability that a user observes an item only depends on its rank. It is parameterized by decreasing position weights: $b_1 \geq ... \geq b_m \geq 0.$ In this model, the utility of user $i$ for a ranking policy $P$ is defined as: $u_i(P) = \sum_{1\leq j,k \leq m} b_k \rho_{ij} P_{ijk}$.
The goal of the recommender system is to find a ranking policy $P$ that maximizes a recommendation objective $f(P):$ $\max_{P \in \mathcal{P}} f(P).$ The traditional objective is the average user utility $f(P) = \frac{1}{n}\sum_{i=1}^n u_i(P),$ which is optimized by sorting $\rho_{ij}$ descendingly for each user $i$ \citep{robertson1977probability}.

In addition to utility, we consider the fairness of rankings for users, defined as balancing the exposure of items across user groups (e.g., ``a job ad should be equally exposed to men and women'') \citep{imana2021auditing,asplund2020auditing,bogen2023towards,usunier2022fast}. Following \citep{usunier2022fast}, we define the exposure of an item $j$ to a user $i$ as $e_{j|i}(P) = \sum_{k=1}^m b_k P_{ijk},\,$ and the exposure of an item $j$ to a user group $g \in \mathcal{G}$ as the average $\mu_{j|g}(P) = \frac{1}{|g|} \sum_{i' \in g} e_{j|i'}(P).$
The user fairness criterion of \emph{balanced exposure} requires that for all items $j,$ the exposures of $j$ to all user groups $(\mu_{j|g}(P))_{g \in \mathcal{G}}$ should be as equal as possible. As in \citep{usunier2022fast}, we formalize it by measuring the ground-truth group unfairness with the standard deviation, averaged over items $j \in \{1,\ldots,m\}$:
\begin{align}\label{eq:ineq-ranking}
\begin{split}
    \frac{1}{m}&\sum_{j=1}^m \tilde{F}\left(\left(\mu_{j|g_i}(P)\right)_{i=1}^n\right)\\ 
    \text{where }\tilde{F}(\Vec{y}) &:= \sqrt{\eta + \sum_{i=1}^n \left(y_i - \frac{1}{n}\sum_{i'=1}^n y_{i'}\right)^2}
\end{split}
\end{align}
$\eta > 0$ is a small smoothing parameter to avoid infinite derivatives at $0$. $\tilde{F}$ is not strictly an inequality measure because it is not normalized by the average outcome. We remove the normalization to make it convex, 
which does not impact our analysis because the sum of exposures is constant (i.e., $\forall j, \, \forall P \in \mathcal{P}, \, \sum_{i=1}^n e_{j|i}(P) = n \| b\|_1).$ Its weighted version $\tilde{F}(\Vec{y},\Vec{w})$ is defined as in Eq. \eqref{eq:weighted_entropy}.

Given a similarity kernel $K,$ we now define the fairness-aware recommendation objective, which is a trade-off between average user utility and unfairness, controlled by a parameter $\beta >0$:
\begin{align}\label{eq:rankingobj-kernel}
\begin{split}
    f_K(P) &= \frac{1}{n} \sum_{i=1}^n u_i(P) - \beta \frac{1}{m}\sum_{j=1}^m \tilde{F}(A(K,e_{j}(P)), K \Vec{1})\\
    &\text{where } e_j(P) = (e_{j|i})_{i=1}^n.
\end{split}
\end{align}

If sensitive group labels $\mathcal{G}$ are available, the recommender system maximizes $\max_{P \in \mathcal{P}} f_{K^*}(P)$, where $K^*$ is the ground-truth kernel, which is equivalent to using the ground-truth group measure \eqref{eq:ineq-ranking} as penalty.
When group labels are not available, we construct a similarity kernel $K$ from the sole network $G$ and solve for $\max_{P \in \mathcal{P}} f_{K}(P).$

We compare the $\max_{P \in \mathcal{P}} f_{K}(P)$ objective with another group-free objective, which simply consists in measuring unfairness at the level of individual users -- this is also equivalent to using $f_K(P)$ when $K$ is the identity matrix $I_n$:
\begin{align}\label{eq:ranking-individual}
    \quad f_{I_n}(P) = \frac{1}{n}\sum_{i=1}^n u_i(P) - \beta \frac{1}{m} \sum_{j=1}^m \tilde{F}(e_{j}(P))
\end{align} 

\subsubsection{Experiments}
We use the Lastfm-Asia and Deezer-Europe datasets \citep{feather}, which both provide user preference data. Following the protocol of \citep{patro2020fairrec,do2021two,wang2021user}, we generate a full user-item preference matrix $(\rho_{ij})_{i,j}$ using a standard matrix completion algorithm (more details are given in App. \ref{app:ranking}). Note that the Deezer-Europe network is much sparser than the Lastfm-Asia network and has much lower assortativity, i.e., lower homophily with respect to the sensitive attribute.

\begin{figure}[t]
    \centering
\includegraphics[width=0.45\linewidth]{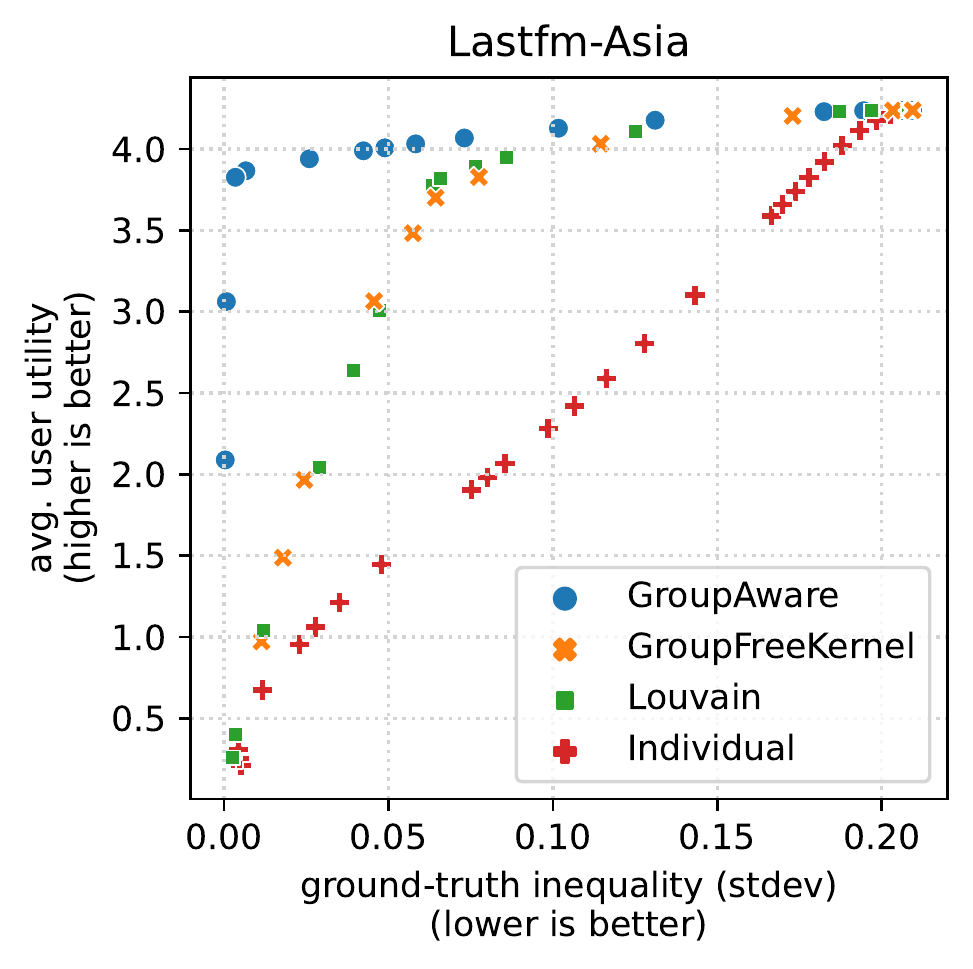}
\includegraphics[width=0.45\linewidth]{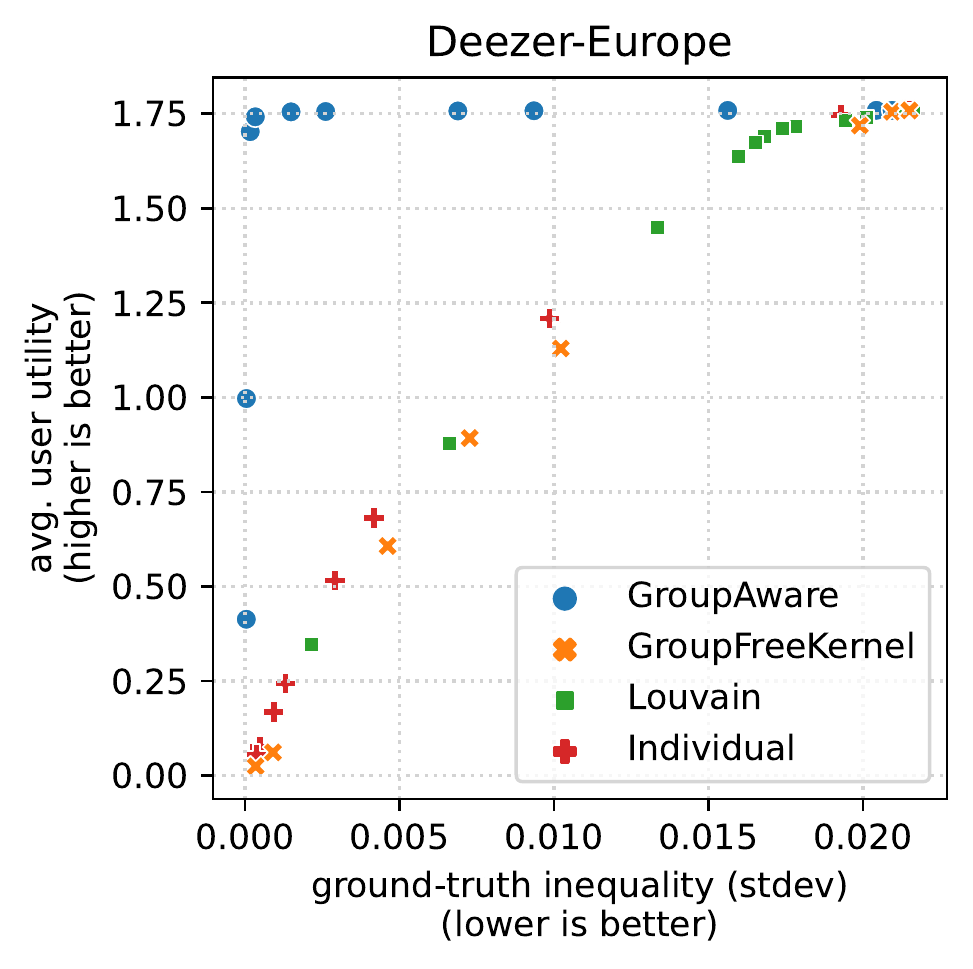}
    \caption[]{\textbf{Recommendation} Utility-fairness trade-offs on (Left) Lastfm-Asia and (Right) Deezer-Europe with our method (\lapker), individual fairness (\indiv), standard group fairness (\groupaware), and community detection (\louvain).
    }\label{fig:main-rk}
\end{figure}

\paragraph{Methods} Our method uses a kernel $K$ (\lapker) constructed from the social network following subsection \ref{sec:kernel-construction}, for various values of embedding dimension $d.$ We generate rankings by solving $\max_{P \in \mathcal{P}} f_K(P)$ using the Frank-Wolfe method for ranking of \citep{do2021two}. We compare them to the rankings obtained from the ideal objective $f_{K^*}$ (\groupaware) (similar to \citep{usunier2022fast}), and the individual fairness baseline $f_{I_n}$ (\indiv). We also consider a community detection baseline (\louvain) which uses the same Louvain community detection kernel described in subsection \ref{sec:influence-max}.
We chose $\eta = 0.1$. The resulting rankings obtained are evaluated by measuring the trade-offs between average user utility and fairness w.r.t. the ground-truth user groups when varying $\beta$ in all objectives.

\paragraph{Results} 
Figure \ref{fig:main-rk} shows the trade-offs obtained by the different methods on Lastfm-Asia (left) and Deezer-Europe (right). Here we show the results obtained with the best parameters for \lapker and \louvain. Experiments with other parameters and community detection baselines are presented in Appendix \ref{app:ranking}. The $y$-axis represents group unfairness, measured by the ground-truth measure
as specified in Equation~\autoref{eq:ineq-ranking}.

On the Lastfm-Asia dataset, our \lapker method achieves trade-offs in-between the \groupaware approach, which requires group labels, and the \indiv fairness approach. As expected, enforcing equality across all individual users instead of groups of users incurs a large cost for total user utility. Our kernel approach is able to mitigate this cost while reducing between-group inequality, without the need for group labels. Our \lapker obtains similar results to the \louvain community detection baseline. On this particular social network, the Louvain algorithm is very good at identifying the ground-truth groups because the network exhibits strong homophily on country. Although our continuous kernel does not perform a hard assignment from users to communities, and hence remains group-free, it still obtains results that are on par with \louvain.

On the Deezer-Europe dataset, both methods based on the social network, i.e. the \louvain community detection baseline and our \lapker approach, obtain poor trade-offs between user utility and inequality. They are no better than simply minimizing individual-level inequalities (\indiv) and very far below the results obtained when knowing the gender group labels (\groupaware). This is because the underlying network is too sparse to extract any information from it: users make very few connections on Deezer and mostly use it to listen to music. This stresses the importance of network structure for our method (and any network-based method). Nonetheless, given this disassortative network, our method does also not degrade performance compared to the individual-level baseline.
\section{Related work and Discussion}\label{sec:discussion}
\paragraph{Fairness and networks}

Several past works have utilized networks to balance outcomes among groups; in contrast to our work, these works generally assume access to sensitive attributes and additionally consider inequality relative to network position. For instance, \citet{boyd2014networked} emphasizes the need to \emph{``also develop legal, social, and ethical models that intentionally account for networks, not just groups and individuals.''}.
Similarly, \citet{mehrotra2022revisiting} introduces the notions of intra-network and inter-group fairness. Intra-network fairness measures inequality in outcomes among groups, where outcomes are weighted by network position. 
On the other hand, inter-group fairness evaluates diversity in \emph{interactions} among groups, e.g., the diversity of connections in a hiring network relative to protected groups. We discuss fairness and network measures in more detail in Appendix \ref{sec:additional-related-work}. 

\paragraph{Group fairness and proxies} When group labels are unavailable, observed proxy labels are often used to approximate group fairness measures. Proxy attributes have been used to measure racial disparities in various domains such as hiring \citep{ahmed2013gay,bertrand2004emily}, ad delivery \citep{sweeney2013discrimination}, finance \citep{bureau2014using,awasthi2021evaluating} and healthcare \citep{fremont2005use,brown2016using}. 
The use of proxy attributes is also common in the economic literature on field experiments for estimating discrimination (see \citet{bertrand2017field} for an overview).
Common proxy attributes include name and address as proxies for gender and race \citep{ahmed2013gay,bertrand2004emily,milkman2012temporal}, as well as past activity in affinity groups 
as a proxy for perceived sexual orientation \citep{ahmed2013gay} and religion \citep{wright2013religious}. 
In the algorithmic fairness literature, several methods aim to incorporate such proxy labels \citep{gupta2018proxy,diana2022multiaccurate} or show under which theoretical assumptions it is possible to recover the ground truth fairness metric values \citep{chen2019fairness,awasthi2020equalized,ghazimatin2022measuring}.
While previously listed proxy measurements assign individuals to binary or categorical groups, social networks allow for a more nuanced view of group membership \citep{atwood2019fair}, since individuals may belong to multiple, overlapping communities \citep{palla2005uncovering}, with complex hierarchies \citep{ronhovde2009multiresolution}. 

An important concern with proxy labels is that inferring unobserved characteristics from them poses privacy risks, as well as dignity concerns \citep{andrus2021we,chun2018queerying}, especially for vulnerable groups, such as queer or religious communities \citep{tomasev2021fairness}. Moreover, many prominent demographic classes, such as race and gender, are cultural constructs that vary over time \citep{jacobs2021measurement}. As emphasized by \citet{jacobs2021measurement}, operationalizing those constructs by measuring proxies may harm marginalized individuals and perpetuate stereotypes \citep{tomasev2021fairness,abbasi2019fairness,chun2018queerying}.
Similarly for social networks, inferring discrete groups by community detection can cause privacy breaches such as
attribute disclosure \citep{zheleva2012privacy} and de-anonymization \citep{nilizadeh2014community}. Since our approach does not assign individuals to discrete immutable categories, it avoids the privacy and stereotyping issues posed by discrete proxies or by partitioning the network. However, like proxy measures, it may lead to poor estimation of the true fairness metric value
\citep{zhang2018assessing,chen2019fairness,baines201411,ghazimatin2022measuring}. 

\paragraph{Limitations}
Our approach is limited to settings where a social network is observable and exhibits homophily for an unobservable attribute of interest.  Moreover, even when such a network is available, it is important to keep in mind that graph datasets are biased by how they are collected and sampled and how links are defined \citep{olteanu2019social,gonzalez2014assessing,de2010does}. These linking biases \citep{olteanu2019social} affect the resulting network-based fairness measures. 

In addition to the limitations and biases of graph data, our method has various practical limitations in real-world use cases. Here we detail three such limitations: First, our approach requires discernable community structure in the network which might not be present in a given network or it might be difficult to infer the correct structure in case of noisy observations (e.g., missing or wrong links). 
Second, our approach relies strongly on the assumption that the relevant sensitive attributes partake in the homophilous network formation process that creates the observed network. While our approach can tolerate some confounding non-sensitive attributes, it relies on the relevant attributes being sufficiently involved in the link-formation process. Results from social science indicate that race and ethnicity, age, religion, education, occupation, and gender are among the strongest sources of homophily in personal environments and roughly in that order~\citep{mcpherson2001birds}. However, these properties vary across networks and environments --- for instance, in heterophilous networks --- and can not be blindly assumed to hold.
Finally, a third limitation is that our measure may require multiple rounds of tuning the embedding algorithm, the number of embedding dimensions, as well as the kernel similarity function, given that the optimal configuration is graph specific.
\section{Conclusion}
In this work, we proposed a novel approach to algorithmic fairness which is based on homophily in social networks to define group fairness without demographics. We ground our approach in well-known results from social science and, furthermore, propose a group-free measure of group fairness based on decomposable inequality metrics. In comparison to naive approaches based on community detection as well as other approaches to group fairness without demographics, our approach does not infer group memberships of individuals at any step of algorithm. Experimentally, we demonstrate the effectiveness and versatility of our approach across three diverse tasks, i.e., classification, maximizing information access, and recommendation. Moreover, even as we remain group-free, our approach is able to yield comparable or better results compared to community detection as our similarity-based approach is able to avoid pitfalls of discrete groups. While social network data might not be available in many scenarios, we believe that our approach is a promising and privacy-respecting approach to group fairness for internet platforms where such data is present. 

\bibliographystyle{ACM-Reference-Format}
\bibliography{references}


\begin{thebibliography}{120}


\ifx \showCODEN    \undefined \def \showCODEN     #1{\unskip}     \fi
\ifx \showDOI      \undefined \def \showDOI       #1{#1}\fi
\ifx \showISBNx    \undefined \def \showISBNx     #1{\unskip}     \fi
\ifx \showISBNxiii \undefined \def \showISBNxiii  #1{\unskip}     \fi
\ifx \showISSN     \undefined \def \showISSN      #1{\unskip}     \fi
\ifx \showLCCN     \undefined \def \showLCCN      #1{\unskip}     \fi
\ifx \shownote     \undefined \def \shownote      #1{#1}          \fi
\ifx \showarticletitle \undefined \def \showarticletitle #1{#1}   \fi
\ifx \showURL      \undefined \def \showURL       {\relax}        \fi
\providecommand\bibfield[2]{#2}
\providecommand\bibinfo[2]{#2}
\providecommand\natexlab[1]{#1}
\providecommand\showeprint[2][]{arXiv:#2}

\bibitem[Abbasi et~al\mbox{.}(2019)]%
        {abbasi2019fairness}
\bibfield{author}{\bibinfo{person}{Mohsen Abbasi}, \bibinfo{person}{Sorelle~A
  Friedler}, \bibinfo{person}{Carlos Scheidegger}, {and}
  \bibinfo{person}{Suresh Venkatasubramanian}.}
  \bibinfo{year}{2019}\natexlab{}.
\newblock \showarticletitle{Fairness in representation: quantifying
  stereotyping as a representational harm}. In
  \bibinfo{booktitle}{\emph{Proceedings of the 2019 SIAM International
  Conference on Data Mining}}. SIAM, \bibinfo{pages}{801--809}.
\newblock


\bibitem[Adamic and Glance(2005)]%
        {polblogs}
\bibfield{author}{\bibinfo{person}{Lada~A. Adamic} {and}
  \bibinfo{person}{Natalie Glance}.} \bibinfo{year}{2005}\natexlab{}.
\newblock \showarticletitle{The Political Blogosphere and the 2004 U.S.
  Election: Divided They Blog}. In \bibinfo{booktitle}{\emph{Proceedings of the
  3rd International Workshop on Link Discovery}}
  \emph{(\bibinfo{series}{LinkKDD '05})}. \bibinfo{publisher}{Association for
  Computing Machinery}, \bibinfo{address}{New York, NY, USA},
  \bibinfo{pages}{36–43}.
\newblock


\bibitem[Ahmed et~al\mbox{.}(2013b)]%
        {ahmed2013distributed}
\bibfield{author}{\bibinfo{person}{Amr Ahmed}, \bibinfo{person}{Nino
  Shervashidze}, \bibinfo{person}{Shravan Narayanamurthy},
  \bibinfo{person}{Vanja Josifovski}, {and} \bibinfo{person}{Alexander~J.
  Smola}.} \bibinfo{year}{2013}\natexlab{b}.
\newblock \showarticletitle{Distributed Large-Scale Natural Graph
  Factorization} \emph{(\bibinfo{series}{WWW '13})}.
  \bibinfo{publisher}{Association for Computing Machinery},
  \bibinfo{address}{New York, NY, USA}, \bibinfo{pages}{37–48}.
\newblock
\showISBNx{9781450320351}
\urldef\tempurl%
\url{https://doi.org/10.1145/2488388.2488393}
\showDOI{\tempurl}


\bibitem[Ahmed et~al\mbox{.}(2013a)]%
        {ahmed2013gay}
\bibfield{author}{\bibinfo{person}{Ali~M Ahmed}, \bibinfo{person}{Lina
  Andersson}, {and} \bibinfo{person}{Mats Hammarstedt}.}
  \bibinfo{year}{2013}\natexlab{a}.
\newblock \showarticletitle{Are gay men and lesbians discriminated against in
  the hiring process?}
\newblock \bibinfo{journal}{\emph{Southern Economic Journal}}
  \bibinfo{volume}{79}, \bibinfo{number}{3} (\bibinfo{year}{2013}),
  \bibinfo{pages}{565--585}.
\newblock


\bibitem[Allison(1978)]%
        {allison1978measures}
\bibfield{author}{\bibinfo{person}{Paul~D Allison}.}
  \bibinfo{year}{1978}\natexlab{}.
\newblock \showarticletitle{Measures of inequality}.
\newblock \bibinfo{journal}{\emph{American sociological review}}
  (\bibinfo{year}{1978}), \bibinfo{pages}{865--880}.
\newblock


\bibitem[Amini et~al\mbox{.}(2013)]%
        {amini2013pseudo}
\bibfield{author}{\bibinfo{person}{Arash~A. Amini}, \bibinfo{person}{Aiyou
  Chen}, \bibinfo{person}{Peter~J. Bickel}, {and} \bibinfo{person}{Elizaveta
  Levina}.} \bibinfo{year}{2013}\natexlab{}.
\newblock \showarticletitle{PSEUDO-LIKELIHOOD METHODS FOR COMMUNITY DETECTION
  IN LARGE SPARSE NETWORKS}.
\newblock \bibinfo{journal}{\emph{The Annals of Statistics}}
  \bibinfo{volume}{41}, \bibinfo{number}{4} (\bibinfo{year}{2013}),
  \bibinfo{pages}{2097--2122}.
\newblock
\urldef\tempurl%
\url{http://www.jstor.org/stable/23566541}
\showURL{%
\tempurl}


\bibitem[Andrus et~al\mbox{.}(2021)]%
        {andrus2021we}
\bibfield{author}{\bibinfo{person}{McKane Andrus}, \bibinfo{person}{Elena
  Spitzer}, \bibinfo{person}{Jeffrey Brown}, {and} \bibinfo{person}{Alice
  Xiang}.} \bibinfo{year}{2021}\natexlab{}.
\newblock \showarticletitle{What We Can't Measure, We Can't Understand:
  Challenges to Demographic Data Procurement in the Pursuit of Fairness}. In
  \bibinfo{booktitle}{\emph{Proceedings of the 2021 ACM Conference on Fairness,
  Accountability, and Transparency}}. \bibinfo{pages}{249--260}.
\newblock


\bibitem[Andrus and Villeneuve(2022a)]%
        {andrus2022demographic}
\bibfield{author}{\bibinfo{person}{McKane Andrus} {and} \bibinfo{person}{Sarah
  Villeneuve}.} \bibinfo{year}{2022}\natexlab{a}.
\newblock \showarticletitle{Demographic-reliant algorithmic fairness:
  characterizing the risks of demographic data collection in the pursuit of
  fairness}. In \bibinfo{booktitle}{\emph{2022 ACM Conference on Fairness,
  Accountability, and Transparency}}. \bibinfo{pages}{1709--1721}.
\newblock


\bibitem[Andrus and Villeneuve(2022b)]%
        {demographic-risks}
\bibfield{author}{\bibinfo{person}{McKane Andrus} {and} \bibinfo{person}{Sarah
  Villeneuve}.} \bibinfo{year}{2022}\natexlab{b}.
\newblock \showarticletitle{Demographic-Reliant Algorithmic Fairness:
  Characterizing the Risks of Demographic Data Collection in the Pursuit of
  Fairness}. In \bibinfo{booktitle}{\emph{2022 ACM Conference on Fairness,
  Accountability, and Transparency}} \emph{(\bibinfo{series}{FAccT '22})}.
  \bibinfo{pages}{1709–1721}.
\newblock


\bibitem[Asplund et~al\mbox{.}(2020)]%
        {asplund2020auditing}
\bibfield{author}{\bibinfo{person}{Joshua Asplund}, \bibinfo{person}{Motahhare
  Eslami}, \bibinfo{person}{Hari Sundaram}, \bibinfo{person}{Christian
  Sandvig}, {and} \bibinfo{person}{Karrie Karahalios}.}
  \bibinfo{year}{2020}\natexlab{}.
\newblock \showarticletitle{Auditing race and gender discrimination in online
  housing markets}. In \bibinfo{booktitle}{\emph{Proceedings of the
  International AAAI Conference on Web and Social Media}},
  Vol.~\bibinfo{volume}{14}. \bibinfo{pages}{24--35}.
\newblock


\bibitem[Atwood et~al\mbox{.}(2019)]%
        {atwood2019fair}
\bibfield{author}{\bibinfo{person}{James Atwood}, \bibinfo{person}{Hansa
  Srinivasan}, \bibinfo{person}{Yoni Halpern}, {and} \bibinfo{person}{David
  Sculley}.} \bibinfo{year}{2019}\natexlab{}.
\newblock \showarticletitle{Fair treatment allocations in social networks}.
\newblock \bibinfo{journal}{\emph{arXiv preprint arXiv:1911.05489}}
  (\bibinfo{year}{2019}).
\newblock


\bibitem[Awasthi et~al\mbox{.}(2021)]%
        {awasthi2021evaluating}
\bibfield{author}{\bibinfo{person}{Pranjal Awasthi}, \bibinfo{person}{Alex
  Beutel}, \bibinfo{person}{Matth{\"a}us Kleindessner}, \bibinfo{person}{Jamie
  Morgenstern}, {and} \bibinfo{person}{Xuezhi Wang}.}
  \bibinfo{year}{2021}\natexlab{}.
\newblock \showarticletitle{Evaluating fairness of machine learning models
  under uncertain and incomplete information}. In
  \bibinfo{booktitle}{\emph{Proceedings of the 2021 ACM Conference on Fairness,
  Accountability, and Transparency}}. \bibinfo{pages}{206--214}.
\newblock


\bibitem[Awasthi et~al\mbox{.}(2020)]%
        {awasthi2020equalized}
\bibfield{author}{\bibinfo{person}{Pranjal Awasthi},
  \bibinfo{person}{Matth{\"a}us Kleindessner}, {and} \bibinfo{person}{Jamie
  Morgenstern}.} \bibinfo{year}{2020}\natexlab{}.
\newblock \showarticletitle{Equalized odds postprocessing under imperfect group
  information}. In \bibinfo{booktitle}{\emph{International Conference on
  Artificial Intelligence and Statistics}}. PMLR, \bibinfo{pages}{1770--1780}.
\newblock


\bibitem[Baines and Courchane(2014)]%
        {baines201411}
\bibfield{author}{\bibinfo{person}{Arthur~P Baines} {and}
  \bibinfo{person}{Marsha~J Courchane}.} \bibinfo{year}{2014}\natexlab{}.
\newblock \bibinfo{title}{11. Fair Lending: Implications for the Indirect Auto
  Finance Market. https}.
\newblock
\newblock


\bibitem[Barocas et~al\mbox{.}(2019)]%
        {barocas-hardt-narayanan}
\bibfield{author}{\bibinfo{person}{Solon Barocas}, \bibinfo{person}{Moritz
  Hardt}, {and} \bibinfo{person}{Arvind Narayanan}.}
  \bibinfo{year}{2019}\natexlab{}.
\newblock \bibinfo{booktitle}{\emph{Fairness and Machine Learning: Limitations
  and Opportunities}}.
\newblock \bibinfo{publisher}{fairmlbook.org}.
\newblock
\newblock
\shownote{\url{http://www.fairmlbook.org}}.


\bibitem[Barocas and Selbst(2016)]%
        {barocas2016big}
\bibfield{author}{\bibinfo{person}{Solon Barocas} {and}
  \bibinfo{person}{Andrew~D Selbst}.} \bibinfo{year}{2016}\natexlab{}.
\newblock \showarticletitle{Big data's disparate impact}.
\newblock \bibinfo{journal}{\emph{California law review}}
  (\bibinfo{year}{2016}), \bibinfo{pages}{671--732}.
\newblock


\bibitem[Belkin and Niyogi(2003)]%
        {lapeigenmap}
\bibfield{author}{\bibinfo{person}{Mikhail Belkin} {and}
  \bibinfo{person}{Partha Niyogi}.} \bibinfo{year}{2003}\natexlab{}.
\newblock \showarticletitle{Laplacian Eigenmaps for Dimensionality Reduction
  and Data Representation}.
\newblock \bibinfo{journal}{\emph{Neural Computation}} \bibinfo{volume}{15},
  \bibinfo{number}{6} (\bibinfo{year}{2003}), \bibinfo{pages}{1373--1396}.
\newblock


\bibitem[Bertrand and Duflo(2017)]%
        {bertrand2017field}
\bibfield{author}{\bibinfo{person}{Marianne Bertrand} {and}
  \bibinfo{person}{Esther Duflo}.} \bibinfo{year}{2017}\natexlab{}.
\newblock \showarticletitle{Field experiments on discrimination}.
\newblock \bibinfo{journal}{\emph{Handbook of economic field experiments}}
  \bibinfo{volume}{1} (\bibinfo{year}{2017}), \bibinfo{pages}{309--393}.
\newblock


\bibitem[Bertrand and Mullainathan(2004)]%
        {bertrand2004emily}
\bibfield{author}{\bibinfo{person}{Marianne Bertrand} {and}
  \bibinfo{person}{Sendhil Mullainathan}.} \bibinfo{year}{2004}\natexlab{}.
\newblock \showarticletitle{Are Emily and Greg more employable than Lakisha and
  Jamal? A field experiment on labor market discrimination}.
\newblock \bibinfo{journal}{\emph{American economic review}}
  \bibinfo{volume}{94}, \bibinfo{number}{4} (\bibinfo{year}{2004}),
  \bibinfo{pages}{991--1013}.
\newblock


\bibitem[Biega et~al\mbox{.}(2018)]%
        {biega2018equity}
\bibfield{author}{\bibinfo{person}{Asia~J Biega}, \bibinfo{person}{Krishna~P
  Gummadi}, {and} \bibinfo{person}{Gerhard Weikum}.}
  \bibinfo{year}{2018}\natexlab{}.
\newblock \showarticletitle{Equity of attention: Amortizing individual fairness
  in rankings}. In \bibinfo{booktitle}{\emph{The 41st international acm sigir
  conference on research \& development in information retrieval}}.
  \bibinfo{pages}{405--414}.
\newblock


\bibitem[Binns(2020)]%
        {binns2020apparent}
\bibfield{author}{\bibinfo{person}{Reuben Binns}.}
  \bibinfo{year}{2020}\natexlab{}.
\newblock \showarticletitle{On the apparent conflict between individual and
  group fairness}. In \bibinfo{booktitle}{\emph{Proceedings of the 2020
  conference on fairness, accountability, and transparency}}.
  \bibinfo{pages}{514--524}.
\newblock


\bibitem[Blondel et~al\mbox{.}(2008)]%
        {blondel2008fast}
\bibfield{author}{\bibinfo{person}{Vincent~D Blondel},
  \bibinfo{person}{Jean-Loup Guillaume}, \bibinfo{person}{Renaud Lambiotte},
  {and} \bibinfo{person}{Etienne Lefebvre}.} \bibinfo{year}{2008}\natexlab{}.
\newblock \showarticletitle{Fast unfolding of communities in large networks}.
\newblock \bibinfo{journal}{\emph{Journal of statistical mechanics: theory and
  experiment}} \bibinfo{volume}{2008}, \bibinfo{number}{10}
  (\bibinfo{year}{2008}), \bibinfo{pages}{P10008}.
\newblock


\bibitem[Boyd et~al\mbox{.}(2014)]%
        {boyd2014networked}
\bibfield{author}{\bibinfo{person}{Danah Boyd}, \bibinfo{person}{Karen Levy},
  {and} \bibinfo{person}{Alice Marwick}.} \bibinfo{year}{2014}\natexlab{}.
\newblock \showarticletitle{The networked nature of algorithmic
  discrimination}.
\newblock \bibinfo{journal}{\emph{Data and Discrimination: Collected Essays.
  Open Technology Institute}} (\bibinfo{year}{2014}).
\newblock


\bibitem[Brown et~al\mbox{.}(2016)]%
        {brown2016using}
\bibfield{author}{\bibinfo{person}{David~P Brown}, \bibinfo{person}{Caprice
  Knapp}, \bibinfo{person}{Kimberly Baker}, {and} \bibinfo{person}{Meggen
  Kaufmann}.} \bibinfo{year}{2016}\natexlab{}.
\newblock \showarticletitle{Using Bayesian imputation to assess racial and
  ethnic disparities in pediatric performance measures}.
\newblock \bibinfo{journal}{\emph{Health services research}}
  \bibinfo{volume}{51}, \bibinfo{number}{3} (\bibinfo{year}{2016}),
  \bibinfo{pages}{1095--1108}.
\newblock


\bibitem[Bureau(2014)]%
        {bureau2014using}
\bibfield{author}{\bibinfo{person}{Consumer Financial~Protection Bureau}.}
  \bibinfo{year}{2014}\natexlab{}.
\newblock \showarticletitle{Using publicly available information to proxy for
  unidentified race and ethnicity: A methodology and assessment}.
\newblock \bibinfo{journal}{\emph{Washington, DC: CFPB, Summer}}
  (\bibinfo{year}{2014}).
\newblock


\bibitem[Burt(1991)]%
        {burt1991measuring}
\bibfield{author}{\bibinfo{person}{Ronald~S. Burt}.}
  \bibinfo{year}{1991}\natexlab{}.
\newblock \showarticletitle{Measuring age as a structural concept}.
\newblock \bibinfo{journal}{\emph{Social Networks}} \bibinfo{volume}{13},
  \bibinfo{number}{1} (\bibinfo{year}{1991}), \bibinfo{pages}{1--34}.
\newblock
\showISSN{0378-8733}
\urldef\tempurl%
\url{https://doi.org/10.1016/0378-8733(91)90011-H}
\showDOI{\tempurl}


\bibitem[Chen et~al\mbox{.}(2019)]%
        {chen2019fairness}
\bibfield{author}{\bibinfo{person}{Jiahao Chen}, \bibinfo{person}{Nathan
  Kallus}, \bibinfo{person}{Xiaojie Mao}, \bibinfo{person}{Geoffry Svacha},
  {and} \bibinfo{person}{Madeleine Udell}.} \bibinfo{year}{2019}\natexlab{}.
\newblock \showarticletitle{Fairness under unawareness: Assessing disparity
  when protected class is unobserved}. In \bibinfo{booktitle}{\emph{Proceedings
  of the conference on fairness, accountability, and transparency}}.
  \bibinfo{pages}{339--348}.
\newblock


\bibitem[Chun(2018)]%
        {chun2018queerying}
\bibfield{author}{\bibinfo{person}{Wendy Hui~Kyong Chun}.}
  \bibinfo{year}{2018}\natexlab{}.
\newblock \showarticletitle{Queerying homophily}.
\newblock  (\bibinfo{year}{2018}).
\newblock


\bibitem[Clauset et~al\mbox{.}(2004)]%
        {clauset2004finding}
\bibfield{author}{\bibinfo{person}{Aaron Clauset}, \bibinfo{person}{Mark~EJ
  Newman}, {and} \bibinfo{person}{Cristopher Moore}.}
  \bibinfo{year}{2004}\natexlab{}.
\newblock \showarticletitle{Finding community structure in very large
  networks}.
\newblock \bibinfo{journal}{\emph{Physical review E}} \bibinfo{volume}{70},
  \bibinfo{number}{6} (\bibinfo{year}{2004}), \bibinfo{pages}{066111}.
\newblock


\bibitem[Cowell(2011)]%
        {cowell2011measuring}
\bibfield{author}{\bibinfo{person}{Frank Cowell}.}
  \bibinfo{year}{2011}\natexlab{}.
\newblock \bibinfo{booktitle}{\emph{Measuring inequality}}.
\newblock \bibinfo{publisher}{Oxford University Press}.
\newblock


\bibitem[Cowell and Kuga(1981)]%
        {cowell1981inequality}
\bibfield{author}{\bibinfo{person}{Frank~A Cowell} {and}
  \bibinfo{person}{Kiyoshi Kuga}.} \bibinfo{year}{1981}\natexlab{}.
\newblock \showarticletitle{Inequality measurement: an axiomatic approach}.
\newblock \bibinfo{journal}{\emph{European Economic Review}}
  \bibinfo{volume}{15}, \bibinfo{number}{3} (\bibinfo{year}{1981}),
  \bibinfo{pages}{287--305}.
\newblock


\bibitem[Dalton(1920)]%
        {dalton1920measurement}
\bibfield{author}{\bibinfo{person}{Hugh Dalton}.}
  \bibinfo{year}{1920}\natexlab{}.
\newblock \showarticletitle{The measurement of the inequality of incomes}.
\newblock \bibinfo{journal}{\emph{The Economic Journal}} \bibinfo{volume}{30},
  \bibinfo{number}{119} (\bibinfo{year}{1920}), \bibinfo{pages}{348--361}.
\newblock


\bibitem[Datta et~al\mbox{.}(2015)]%
        {datta2015automated}
\bibfield{author}{\bibinfo{person}{Amit Datta}, \bibinfo{person}{Michael~Carl
  Tschantz}, \bibinfo{person}{Anupam Datta}, {and} \bibinfo{person}{.}}
  \bibinfo{year}{2015}\natexlab{}.
\newblock \showarticletitle{Automated experiments on ad privacy settings}.
\newblock \bibinfo{journal}{\emph{Proceedings on Privacy Enhancing
  Technologies}} \bibinfo{volume}{2015}, \bibinfo{number}{1}
  (\bibinfo{year}{2015}), \bibinfo{pages}{92--112}.
\newblock


\bibitem[De~Choudhury et~al\mbox{.}(2010)]%
        {de2010does}
\bibfield{author}{\bibinfo{person}{Munmun De~Choudhury}, \bibinfo{person}{Yu-Ru
  Lin}, \bibinfo{person}{Hari Sundaram}, \bibinfo{person}{Kasim~Selcuk Candan},
  \bibinfo{person}{Lexing Xie}, {and} \bibinfo{person}{Aisling Kelliher}.}
  \bibinfo{year}{2010}\natexlab{}.
\newblock \showarticletitle{How does the data sampling strategy impact the
  discovery of information diffusion in social media?}. In
  \bibinfo{booktitle}{\emph{Fourth international AAAI conference on weblogs and
  social media}}.
\newblock


\bibitem[Deng et~al\mbox{.}(2020)]%
        {deng2020strong}
\bibfield{author}{\bibinfo{person}{Shaofeng Deng}, \bibinfo{person}{Shuyang
  Ling}, {and} \bibinfo{person}{Thomas Strohmer}.}
  \bibinfo{year}{2020}\natexlab{}.
\newblock \showarticletitle{Strong consistency, graph Laplacians, and the
  stochastic block model}.
\newblock \bibinfo{journal}{\emph{arXiv preprint arXiv:2004.09780}}
  (\bibinfo{year}{2020}).
\newblock


\bibitem[Diana et~al\mbox{.}(2022)]%
        {diana2022multiaccurate}
\bibfield{author}{\bibinfo{person}{Emily Diana}, \bibinfo{person}{Wesley Gill},
  \bibinfo{person}{Michael Kearns}, \bibinfo{person}{Krishnaram Kenthapadi},
  \bibinfo{person}{Aaron Roth}, {and} \bibinfo{person}{Saeed
  Sharifi-Malvajerdi}.} \bibinfo{year}{2022}\natexlab{}.
\newblock \showarticletitle{Multiaccurate proxies for downstream fairness}. In
  \bibinfo{booktitle}{\emph{2022 ACM Conference on Fairness, Accountability,
  and Transparency}}. \bibinfo{pages}{1207--1239}.
\newblock


\bibitem[Do et~al\mbox{.}(2021)]%
        {do2021two}
\bibfield{author}{\bibinfo{person}{Virginie Do}, \bibinfo{person}{Sam
  Corbett-Davies}, \bibinfo{person}{Jamal Atif}, {and} \bibinfo{person}{Nicolas
  Usunier}.} \bibinfo{year}{2021}\natexlab{}.
\newblock \showarticletitle{Two-sided fairness in rankings via Lorenz
  dominance}.
\newblock \bibinfo{journal}{\emph{Advances in Neural Information Processing
  Systems}}  \bibinfo{volume}{34} (\bibinfo{year}{2021}),
  \bibinfo{pages}{8596--8608}.
\newblock


\bibitem[Duchin and Murphy(2022)]%
        {Duchin2022}
\bibfield{author}{\bibinfo{person}{Moon Duchin} {and} \bibinfo{person}{James~M.
  Murphy}.} \bibinfo{year}{2022}\natexlab{}.
\newblock \bibinfo{booktitle}{\emph{Explainer: Measuring clustering and
  segregation}}.
\newblock \bibinfo{publisher}{Springer International Publishing},
  \bibinfo{pages}{293--302}.
\newblock
\showISBNx{978-3-319-69161-9}
\urldef\tempurl%
\url{https://doi.org/10.1007/978-3-319-69161-9_15}
\showDOI{\tempurl}


\bibitem[Dwork et~al\mbox{.}(2012)]%
        {dwork2012fairness}
\bibfield{author}{\bibinfo{person}{Cynthia Dwork}, \bibinfo{person}{Moritz
  Hardt}, \bibinfo{person}{Toniann Pitassi}, \bibinfo{person}{Omer Reingold},
  {and} \bibinfo{person}{Richard Zemel}.} \bibinfo{year}{2012}\natexlab{}.
\newblock \showarticletitle{Fairness through awareness}. In
  \bibinfo{booktitle}{\emph{Proceedings of the 3rd innovations in theoretical
  computer science conference}}. \bibinfo{pages}{214--226}.
\newblock


\bibitem[Feldman et~al\mbox{.}(2015)]%
        {feldman2015certifying}
\bibfield{author}{\bibinfo{person}{Michael Feldman}, \bibinfo{person}{Sorelle~A
  Friedler}, \bibinfo{person}{John Moeller}, \bibinfo{person}{Carlos
  Scheidegger}, {and} \bibinfo{person}{Suresh Venkatasubramanian}.}
  \bibinfo{year}{2015}\natexlab{}.
\newblock \showarticletitle{Certifying and removing disparate impact}. In
  \bibinfo{booktitle}{\emph{proceedings of the 21th ACM SIGKDD international
  conference on knowledge discovery and data mining}}.
  \bibinfo{pages}{259--268}.
\newblock


\bibitem[Fish et~al\mbox{.}(2019)]%
        {fish}
\bibfield{author}{\bibinfo{person}{Benjamin Fish}, \bibinfo{person}{Ashkan
  Bashardoust}, \bibinfo{person}{Danah Boyd}, \bibinfo{person}{Sorelle
  Friedler}, \bibinfo{person}{Carlos Scheidegger}, {and}
  \bibinfo{person}{Suresh Venkatasubramanian}.}
  \bibinfo{year}{2019}\natexlab{}.
\newblock \showarticletitle{Gaps in Information Access in Social Networks?}. In
  \bibinfo{booktitle}{\emph{The World Wide Web Conference}}
  \emph{(\bibinfo{series}{WWW '19})}. \bibinfo{pages}{480–490}.
\newblock


\bibitem[Fortunato(2010)]%
        {fortunato2010community}
\bibfield{author}{\bibinfo{person}{Santo Fortunato}.}
  \bibinfo{year}{2010}\natexlab{}.
\newblock \showarticletitle{Community detection in graphs}.
\newblock \bibinfo{journal}{\emph{Physics Reports}} \bibinfo{volume}{486},
  \bibinfo{number}{3} (\bibinfo{year}{2010}), \bibinfo{pages}{75--174}.
\newblock
\showISSN{0370-1573}
\urldef\tempurl%
\url{https://doi.org/10.1016/j.physrep.2009.11.002}
\showDOI{\tempurl}


\bibitem[Fremont et~al\mbox{.}(2005)]%
        {fremont2005use}
\bibfield{author}{\bibinfo{person}{Allen~M Fremont}, \bibinfo{person}{Arlene
  Bierman}, \bibinfo{person}{Steve~L Wickstrom}, \bibinfo{person}{Chloe~E
  Bird}, \bibinfo{person}{Mona Shah}, \bibinfo{person}{Jos{\'e}~J Escarce},
  \bibinfo{person}{Thomas Horstman}, {and} \bibinfo{person}{Thomas Rector}.}
  \bibinfo{year}{2005}\natexlab{}.
\newblock \showarticletitle{Use of geocoding in managed care settings to
  identify quality disparities}.
\newblock \bibinfo{journal}{\emph{Health Affairs}} \bibinfo{volume}{24},
  \bibinfo{number}{2} (\bibinfo{year}{2005}), \bibinfo{pages}{516--526}.
\newblock


\bibitem[Ghazimatin et~al\mbox{.}(2022)]%
        {ghazimatin2022measuring}
\bibfield{author}{\bibinfo{person}{Azin Ghazimatin}, \bibinfo{person}{Matthaus
  Kleindessner}, \bibinfo{person}{Chris Russell}, \bibinfo{person}{Ziawasch
  Abedjan}, {and} \bibinfo{person}{Jacek Golebiowski}.}
  \bibinfo{year}{2022}\natexlab{}.
\newblock \showarticletitle{Measuring fairness of rankings under noisy
  sensitive information}. In \bibinfo{booktitle}{\emph{2022 ACM Conference on
  Fairness, Accountability, and Transparency}}. \bibinfo{pages}{2263--2279}.
\newblock


\bibitem[Girvan and Newman(2002)]%
        {girvan2002community}
\bibfield{author}{\bibinfo{person}{Michelle Girvan} {and}
  \bibinfo{person}{Mark~EJ Newman}.} \bibinfo{year}{2002}\natexlab{}.
\newblock \showarticletitle{Community structure in social and biological
  networks}.
\newblock \bibinfo{journal}{\emph{Proceedings of the national academy of
  sciences}} \bibinfo{volume}{99}, \bibinfo{number}{12} (\bibinfo{year}{2002}),
  \bibinfo{pages}{7821--7826}.
\newblock


\bibitem[Gonz{\'a}lez-Bail{\'o}n et~al\mbox{.}(2014)]%
        {gonzalez2014assessing}
\bibfield{author}{\bibinfo{person}{Sandra Gonz{\'a}lez-Bail{\'o}n},
  \bibinfo{person}{Ning Wang}, \bibinfo{person}{Alejandro Rivero},
  \bibinfo{person}{Javier Borge-Holthoefer}, {and} \bibinfo{person}{Yamir
  Moreno}.} \bibinfo{year}{2014}\natexlab{}.
\newblock \showarticletitle{Assessing the bias in samples of large online
  networks}.
\newblock \bibinfo{journal}{\emph{Social Networks}}  \bibinfo{volume}{38}
  (\bibinfo{year}{2014}), \bibinfo{pages}{16--27}.
\newblock


\bibitem[Grover and Leskovec(2016)]%
        {grover2016node2vec}
\bibfield{author}{\bibinfo{person}{Aditya Grover} {and} \bibinfo{person}{Jure
  Leskovec}.} \bibinfo{year}{2016}\natexlab{}.
\newblock \showarticletitle{Node2vec: Scalable Feature Learning for Networks}.
  In \bibinfo{booktitle}{\emph{Proceedings of the 22nd ACM SIGKDD International
  Conference on Knowledge Discovery and Data Mining}} (San Francisco,
  California, USA) \emph{(\bibinfo{series}{KDD '16})}.
  \bibinfo{publisher}{Association for Computing Machinery},
  \bibinfo{address}{New York, NY, USA}, \bibinfo{pages}{855–864}.
\newblock
\showISBNx{9781450342322}
\urldef\tempurl%
\url{https://doi.org/10.1145/2939672.2939754}
\showDOI{\tempurl}


\bibitem[Gupta et~al\mbox{.}(2018)]%
        {gupta2018proxy}
\bibfield{author}{\bibinfo{person}{Maya Gupta}, \bibinfo{person}{Andrew
  Cotter}, \bibinfo{person}{Mahdi~Milani Fard}, {and} \bibinfo{person}{Serena
  Wang}.} \bibinfo{year}{2018}\natexlab{}.
\newblock \showarticletitle{Proxy fairness}.
\newblock \bibinfo{journal}{\emph{arXiv preprint arXiv:1806.11212}}
  (\bibinfo{year}{2018}).
\newblock


\bibitem[Hamilton et~al\mbox{.}(2017)]%
        {hamilton2017inductive}
\bibfield{author}{\bibinfo{person}{William~L. Hamilton}, \bibinfo{person}{Rex
  Ying}, {and} \bibinfo{person}{Jure Leskovec}.}
  \bibinfo{year}{2017}\natexlab{}.
\newblock \showarticletitle{Inductive Representation Learning on Large Graphs}.
  In \bibinfo{booktitle}{\emph{Proceedings of the 31st International Conference
  on Neural Information Processing Systems}} (Long Beach, California, USA)
  \emph{(\bibinfo{series}{NIPS'17})}. \bibinfo{publisher}{Curran Associates
  Inc.}, \bibinfo{address}{Red Hook, NY, USA}, \bibinfo{pages}{1025–1035}.
\newblock
\showISBNx{9781510860964}


\bibitem[Handcock et~al\mbox{.}(2007)]%
        {handcock2007model}
\bibfield{author}{\bibinfo{person}{Mark~S Handcock}, \bibinfo{person}{Adrian~E
  Raftery}, {and} \bibinfo{person}{Jeremy~M Tantrum}.}
  \bibinfo{year}{2007}\natexlab{}.
\newblock \showarticletitle{Model-based clustering for social networks}.
\newblock \bibinfo{journal}{\emph{Journal of the Royal Statistical Society:
  Series A (Statistics in Society)}} \bibinfo{volume}{170}, \bibinfo{number}{2}
  (\bibinfo{year}{2007}), \bibinfo{pages}{301--354}.
\newblock


\bibitem[Hardt et~al\mbox{.}(2016)]%
        {equality-of-opportunity}
\bibfield{author}{\bibinfo{person}{Moritz Hardt}, \bibinfo{person}{Eric Price},
  \bibinfo{person}{Eric Price}, {and} \bibinfo{person}{Nati Srebro}.}
  \bibinfo{year}{2016}\natexlab{}.
\newblock \showarticletitle{Equality of Opportunity in Supervised Learning}. In
  \bibinfo{booktitle}{\emph{Advances in Neural Information Processing
  Systems}}, Vol.~\bibinfo{volume}{29}. \bibinfo{publisher}{Curran Associates,
  Inc.}
\newblock
\urldef\tempurl%
\url{https://proceedings.neurips.cc/paper/2016/file/9d2682367c3935defcb1f9e247a97c0d-Paper.pdf}
\showURL{%
\tempurl}


\bibitem[Hashimoto et~al\mbox{.}(2018)]%
        {hashimoto2018fairness}
\bibfield{author}{\bibinfo{person}{Tatsunori Hashimoto}, \bibinfo{person}{Megha
  Srivastava}, \bibinfo{person}{Hongseok Namkoong}, {and}
  \bibinfo{person}{Percy Liang}.} \bibinfo{year}{2018}\natexlab{}.
\newblock \showarticletitle{Fairness without demographics in repeated loss
  minimization}. In \bibinfo{booktitle}{\emph{International Conference on
  Machine Learning}}. PMLR, \bibinfo{pages}{1929--1938}.
\newblock


\bibitem[Hoff(2008)]%
        {hoff2007modeling}
\bibfield{author}{\bibinfo{person}{Peter Hoff}.}
  \bibinfo{year}{2008}\natexlab{}.
\newblock \showarticletitle{Modeling homophily and stochastic equivalence in
  symmetric relational data}. In \bibinfo{booktitle}{\emph{Advances in Neural
  Information Processing Systems}},
  \bibfield{editor}{\bibinfo{person}{J.~Platt}, \bibinfo{person}{D.~Koller},
  \bibinfo{person}{Y.~Singer}, {and} \bibinfo{person}{S.~Roweis}} (Eds.),
  Vol.~\bibinfo{volume}{20}. \bibinfo{publisher}{Curran Associates, Inc.}
\newblock
\urldef\tempurl%
\url{https://proceedings.neurips.cc/paper/2007/file/766ebcd59621e305170616ba3d3dac32-Paper.pdf}
\showURL{%
\tempurl}


\bibitem[Hoff et~al\mbox{.}(2002)]%
        {hoff2002latent}
\bibfield{author}{\bibinfo{person}{Peter~D Hoff}, \bibinfo{person}{Adrian~E
  Raftery}, {and} \bibinfo{person}{Mark~S Handcock}.}
  \bibinfo{year}{2002}\natexlab{}.
\newblock \showarticletitle{Latent space approaches to social network
  analysis}.
\newblock \bibinfo{journal}{\emph{Journal of the american Statistical
  association}} \bibinfo{volume}{97}, \bibinfo{number}{460}
  (\bibinfo{year}{2002}), \bibinfo{pages}{1090--1098}.
\newblock


\bibitem[Holland et~al\mbox{.}(1983)]%
        {holland1983stochastic}
\bibfield{author}{\bibinfo{person}{Paul~W. Holland},
  \bibinfo{person}{Kathryn~Blackmond Laskey}, {and} \bibinfo{person}{Samuel
  Leinhardt}.} \bibinfo{year}{1983}\natexlab{}.
\newblock \showarticletitle{Stochastic blockmodels: First steps}.
\newblock \bibinfo{journal}{\emph{Social Networks}} \bibinfo{volume}{5},
  \bibinfo{number}{2} (\bibinfo{year}{1983}), \bibinfo{pages}{109--137}.
\newblock
\showISSN{0378-8733}
\urldef\tempurl%
\url{https://doi.org/10.1016/0378-8733(83)90021-7}
\showDOI{\tempurl}


\bibitem[Holstein et~al\mbox{.}(2019)]%
        {holstein2019improving}
\bibfield{author}{\bibinfo{person}{Kenneth Holstein}, \bibinfo{person}{Jennifer
  Wortman~Vaughan}, \bibinfo{person}{Hal Daum{\'e}~III}, \bibinfo{person}{Miro
  Dudik}, {and} \bibinfo{person}{Hanna Wallach}.}
  \bibinfo{year}{2019}\natexlab{}.
\newblock \showarticletitle{Improving fairness in machine learning systems:
  What do industry practitioners need?}. In
  \bibinfo{booktitle}{\emph{Proceedings of the 2019 CHI conference on human
  factors in computing systems}}. \bibinfo{pages}{1--16}.
\newblock


\bibitem[Hu et~al\mbox{.}(2008)]%
        {hu2008collaborative}
\bibfield{author}{\bibinfo{person}{Yifan Hu}, \bibinfo{person}{Yehuda Koren},
  {and} \bibinfo{person}{Chris Volinsky}.} \bibinfo{year}{2008}\natexlab{}.
\newblock \showarticletitle{Collaborative filtering for implicit feedback
  datasets}. In \bibinfo{booktitle}{\emph{2008 Eighth IEEE international
  conference on data mining}}. Ieee, \bibinfo{pages}{263--272}.
\newblock


\bibitem[Ibarra(1995)]%
        {ibarra1995race}
\bibfield{author}{\bibinfo{person}{Herminia Ibarra}.}
  \bibinfo{year}{1995}\natexlab{}.
\newblock \showarticletitle{Race, Opportunity, and Diversity of Social Circles
  in Managerial Networks}.
\newblock \bibinfo{journal}{\emph{The Academy of Management Journal}}
  \bibinfo{volume}{38}, \bibinfo{number}{3} (\bibinfo{year}{1995}),
  \bibinfo{pages}{673--703}.
\newblock
\showISSN{00014273}
\urldef\tempurl%
\url{http://www.jstor.org/stable/256742}
\showURL{%
\tempurl}


\bibitem[Imana et~al\mbox{.}(2021)]%
        {imana2021auditing}
\bibfield{author}{\bibinfo{person}{Basileal Imana}, \bibinfo{person}{Aleksandra
  Korolova}, {and} \bibinfo{person}{John Heidemann}.}
  \bibinfo{year}{2021}\natexlab{}.
\newblock \showarticletitle{Auditing for discrimination in algorithms
  delivering job ads}. In \bibinfo{booktitle}{\emph{Proceedings of the Web
  Conference 2021}}. \bibinfo{pages}{3767--3778}.
\newblock


\bibitem[Jacobs and Wallach(2021)]%
        {jacobs2021measurement}
\bibfield{author}{\bibinfo{person}{Abigail~Z Jacobs} {and}
  \bibinfo{person}{Hanna Wallach}.} \bibinfo{year}{2021}\natexlab{}.
\newblock \showarticletitle{Measurement and fairness}. In
  \bibinfo{booktitle}{\emph{Proceedings of the 2021 ACM conference on fairness,
  accountability, and transparency}}. \bibinfo{pages}{375--385}.
\newblock


\bibitem[Jagielski et~al\mbox{.}(2019)]%
        {jagielski2019differentially}
\bibfield{author}{\bibinfo{person}{Matthew Jagielski}, \bibinfo{person}{Michael
  Kearns}, \bibinfo{person}{Jieming Mao}, \bibinfo{person}{Alina Oprea},
  \bibinfo{person}{Aaron Roth}, \bibinfo{person}{Saeed Sharifi-Malvajerdi},
  {and} \bibinfo{person}{Jonathan Ullman}.} \bibinfo{year}{2019}\natexlab{}.
\newblock \showarticletitle{Differentially private fair learning}. In
  \bibinfo{booktitle}{\emph{International Conference on Machine Learning}}.
  PMLR, \bibinfo{pages}{3000--3008}.
\newblock


\bibitem[Juarez and Korolova(2022)]%
        {juarez2022you}
\bibfield{author}{\bibinfo{person}{Marc Juarez} {and}
  \bibinfo{person}{Aleksandra Korolova}.} \bibinfo{year}{2022}\natexlab{}.
\newblock \showarticletitle{" You Can't Fix What You Can't Measure": Privately
  Measuring Demographic Performance Disparities in Federated Learning}.
\newblock \bibinfo{journal}{\emph{arXiv preprint arXiv:2206.12183}}
  (\bibinfo{year}{2022}).
\newblock


\bibitem[Kalmijn(1998)]%
        {kalmijn1998intermarriage}
\bibfield{author}{\bibinfo{person}{Matthijs Kalmijn}.}
  \bibinfo{year}{1998}\natexlab{}.
\newblock \showarticletitle{Intermarriage and Homogamy: Causes, Patterns,
  Trends}.
\newblock \bibinfo{journal}{\emph{Annual Review of Sociology}}
  \bibinfo{volume}{24}, \bibinfo{number}{1} (\bibinfo{year}{1998}),
  \bibinfo{pages}{395--421}.
\newblock
\urldef\tempurl%
\url{https://doi.org/10.1146/annurev.soc.24.1.395}
\showDOI{\tempurl}
\showeprint{https://doi.org/10.1146/annurev.soc.24.1.395}
\newblock
\shownote{PMID: 12321971}.


\bibitem[Kempe et~al\mbox{.}(2003)]%
        {kleinberg2003influence}
\bibfield{author}{\bibinfo{person}{David Kempe}, \bibinfo{person}{Jon
  Kleinberg}, {and} \bibinfo{person}{\'{E}va Tardos}.}
  \bibinfo{year}{2003}\natexlab{}.
\newblock \showarticletitle{Maximizing the Spread of Influence through a Social
  Network}. In \bibinfo{booktitle}{\emph{Proceedings of the Ninth ACM SIGKDD
  International Conference on Knowledge Discovery and Data Mining}}
  \emph{(\bibinfo{series}{KDD '03})}. \bibinfo{publisher}{Association for
  Computing Machinery}, \bibinfo{pages}{137–146}.
\newblock
\urldef\tempurl%
\url{https://doi.org/10.1145/956750.956769}
\showDOI{\tempurl}


\bibitem[Kilbertus et~al\mbox{.}(2018)]%
        {kilbertus2018blind}
\bibfield{author}{\bibinfo{person}{Niki Kilbertus}, \bibinfo{person}{Adri{\`a}
  Gasc{\'o}n}, \bibinfo{person}{Matt Kusner}, \bibinfo{person}{Michael Veale},
  \bibinfo{person}{Krishna Gummadi}, {and} \bibinfo{person}{Adrian Weller}.}
  \bibinfo{year}{2018}\natexlab{}.
\newblock \showarticletitle{Blind justice: Fairness with encrypted sensitive
  attributes}. In \bibinfo{booktitle}{\emph{International Conference on Machine
  Learning}}. PMLR, \bibinfo{pages}{2630--2639}.
\newblock


\bibitem[Kipf and Welling(2017)]%
        {kipf2016semi}
\bibfield{author}{\bibinfo{person}{Thomas~N. Kipf} {and} \bibinfo{person}{Max
  Welling}.} \bibinfo{year}{2017}\natexlab{}.
\newblock \showarticletitle{Semi-Supervised Classification with Graph
  Convolutional Networks}. In \bibinfo{booktitle}{\emph{International
  Conference on Learning Representations}}.
\newblock


\bibitem[Kossinets and Watts(2009)]%
        {kossinets2009origins}
\bibfield{author}{\bibinfo{person}{Gueorgi Kossinets} {and}
  \bibinfo{person}{Duncan J. Watts}.} \bibinfo{year}{2009}\natexlab{}.
\newblock \showarticletitle{Origins of Homophily in an Evolving Social
  Network}.
\newblock \bibinfo{journal}{\emph{Amer. J. Sociology}} \bibinfo{volume}{115},
  \bibinfo{number}{2} (\bibinfo{year}{2009}), \bibinfo{pages}{405--450}.
\newblock
\urldef\tempurl%
\url{https://doi.org/10.1086/599247}
\showDOI{\tempurl}
\showeprint{https://doi.org/10.1086/599247}


\bibitem[Lahoti et~al\mbox{.}(2020)]%
        {lahoti2020fairness}
\bibfield{author}{\bibinfo{person}{Preethi Lahoti}, \bibinfo{person}{Alex
  Beutel}, \bibinfo{person}{Jilin Chen}, \bibinfo{person}{Kang Lee},
  \bibinfo{person}{Flavien Prost}, \bibinfo{person}{Nithum Thain},
  \bibinfo{person}{Xuezhi Wang}, {and} \bibinfo{person}{Ed Chi}.}
  \bibinfo{year}{2020}\natexlab{}.
\newblock \showarticletitle{Fairness without demographics through adversarially
  reweighted learning}.
\newblock \bibinfo{journal}{\emph{Advances in neural information processing
  systems}}  \bibinfo{volume}{33} (\bibinfo{year}{2020}),
  \bibinfo{pages}{728--740}.
\newblock


\bibitem[Lambrecht and Tucker(2019)]%
        {lambrecht2019algorithmic}
\bibfield{author}{\bibinfo{person}{Anja Lambrecht} {and}
  \bibinfo{person}{Catherine Tucker}.} \bibinfo{year}{2019}\natexlab{}.
\newblock \showarticletitle{Algorithmic bias? An empirical study of apparent
  gender-based discrimination in the display of STEM career ads}.
\newblock \bibinfo{journal}{\emph{Management science}} \bibinfo{volume}{65},
  \bibinfo{number}{7} (\bibinfo{year}{2019}), \bibinfo{pages}{2966--2981}.
\newblock


\bibitem[Laumann(1973)]%
        {laumann1973bonds}
\bibfield{author}{\bibinfo{person}{Edward~O Laumann}.}
  \bibinfo{year}{1973}\natexlab{}.
\newblock \bibinfo{booktitle}{\emph{Bonds of pluralism: The form and substance
  of urban social networks}}.
\newblock \bibinfo{publisher}{New York: J. Wiley}.
\newblock


\bibitem[Lazarsfeld et~al\mbox{.}(1954)]%
        {lazarsfeld1954friendship}
\bibfield{author}{\bibinfo{person}{Paul~F Lazarsfeld},
  \bibinfo{person}{Robert~K Merton}, {et~al\mbox{.}}}
  \bibinfo{year}{1954}\natexlab{}.
\newblock \showarticletitle{Friendship as a social process: A substantive and
  methodological analysis}.
\newblock \bibinfo{journal}{\emph{Freedom and control in modern society}}
  \bibinfo{volume}{18}, \bibinfo{number}{1} (\bibinfo{year}{1954}),
  \bibinfo{pages}{18--66}.
\newblock


\bibitem[Lee et~al\mbox{.}(2014)]%
        {lee2014multiway}
\bibfield{author}{\bibinfo{person}{James~R. Lee}, \bibinfo{person}{Shayan~Oveis
  Gharan}, {and} \bibinfo{person}{Luca Trevisan}.}
  \bibinfo{year}{2014}\natexlab{}.
\newblock \showarticletitle{Multiway Spectral Partitioning and Higher-Order
  Cheeger Inequalities}.
\newblock \bibinfo{journal}{\emph{J. ACM}} \bibinfo{volume}{61},
  \bibinfo{number}{6}, Article \bibinfo{articleno}{37} (\bibinfo{date}{dec}
  \bibinfo{year}{2014}), \bibinfo{numpages}{30}~pages.
\newblock
\showISSN{0004-5411}
\urldef\tempurl%
\url{https://doi.org/10.1145/2665063}
\showDOI{\tempurl}


\bibitem[Louch(2000)]%
        {louch2000personal}
\bibfield{author}{\bibinfo{person}{Hugh Louch}.}
  \bibinfo{year}{2000}\natexlab{}.
\newblock \showarticletitle{Personal network integration: transitivity and
  homophily in strong-tie relations}.
\newblock \bibinfo{journal}{\emph{Social Networks}} \bibinfo{volume}{22},
  \bibinfo{number}{1} (\bibinfo{year}{2000}), \bibinfo{pages}{45--64}.
\newblock
\showISSN{0378-8733}
\urldef\tempurl%
\url{https://doi.org/10.1016/S0378-8733(00)00015-0}
\showDOI{\tempurl}


\bibitem[Marsden(1987)]%
        {marsden1987core}
\bibfield{author}{\bibinfo{person}{Peter~V. Marsden}.}
  \bibinfo{year}{1987}\natexlab{}.
\newblock \showarticletitle{Core Discussion Networks of Americans}.
\newblock \bibinfo{journal}{\emph{American Sociological Review}}
  \bibinfo{volume}{52}, \bibinfo{number}{1} (\bibinfo{year}{1987}),
  \bibinfo{pages}{122--131}.
\newblock
\showISSN{00031224}
\urldef\tempurl%
\url{http://www.jstor.org/stable/2095397}
\showURL{%
\tempurl}


\bibitem[Marsden(1988)]%
        {marsden1988homogeneity}
\bibfield{author}{\bibinfo{person}{Peter~V Marsden}.}
  \bibinfo{year}{1988}\natexlab{}.
\newblock \showarticletitle{Homogeneity in confiding relations}.
\newblock \bibinfo{journal}{\emph{Social networks}} \bibinfo{volume}{10},
  \bibinfo{number}{1} (\bibinfo{year}{1988}), \bibinfo{pages}{57--76}.
\newblock


\bibitem[Mayhew et~al\mbox{.}(1995)]%
        {mayhew1995sex}
\bibfield{author}{\bibinfo{person}{Bruce~H. Mayhew}, \bibinfo{person}{J.~Miller
  McPherson}, \bibinfo{person}{Thomas Rotolo}, {and} \bibinfo{person}{Lynn
  Smith-Lovin}.} \bibinfo{year}{1995}\natexlab{}.
\newblock \showarticletitle{Sex and Race Homogeneity in Naturally Occurring
  Groups}.
\newblock \bibinfo{journal}{\emph{Social Forces}}  \bibinfo{volume}{74}
  (\bibinfo{year}{1995}), \bibinfo{pages}{15--52}.
\newblock


\bibitem[McPherson and Smith-Lovin(1986)]%
        {mcpherson1986sex}
\bibfield{author}{\bibinfo{person}{J.~Miller McPherson} {and}
  \bibinfo{person}{Lynn Smith-Lovin}.} \bibinfo{year}{1986}\natexlab{}.
\newblock \showarticletitle{Sex Segregation in Voluntary Associations}.
\newblock \bibinfo{journal}{\emph{American Sociological Review}}
  \bibinfo{volume}{51} (\bibinfo{year}{1986}), \bibinfo{pages}{61}.
\newblock


\bibitem[McPherson and Smith-Lovin(1987)]%
        {mcpherson1987homophily}
\bibfield{author}{\bibinfo{person}{J.~Miller McPherson} {and}
  \bibinfo{person}{Lynn Smith-Lovin}.} \bibinfo{year}{1987}\natexlab{}.
\newblock \showarticletitle{Homophily in Voluntary Organizations: Status
  Distance and the Composition of Face-to-Face Groups}.
\newblock \bibinfo{journal}{\emph{American Sociological Review}}
  \bibinfo{volume}{52}, \bibinfo{number}{3} (\bibinfo{year}{1987}),
  \bibinfo{pages}{370--379}.
\newblock
\showISSN{00031224}
\urldef\tempurl%
\url{http://www.jstor.org/stable/2095356}
\showURL{%
\tempurl}


\bibitem[McPherson et~al\mbox{.}(2001)]%
        {mcpherson2001birds}
\bibfield{author}{\bibinfo{person}{Miller McPherson}, \bibinfo{person}{Lynn
  Smith-Lovin}, {and} \bibinfo{person}{James~M Cook}.}
  \bibinfo{year}{2001}\natexlab{}.
\newblock \showarticletitle{Birds of a feather: Homophily in social networks}.
\newblock \bibinfo{journal}{\emph{Annual review of sociology}}
  (\bibinfo{year}{2001}), \bibinfo{pages}{415--444}.
\newblock


\bibitem[Mehrotra and Celis(2021)]%
        {mehrotra2021mitigating}
\bibfield{author}{\bibinfo{person}{Anay Mehrotra} {and}
  \bibinfo{person}{L~Elisa Celis}.} \bibinfo{year}{2021}\natexlab{}.
\newblock \showarticletitle{Mitigating bias in set selection with noisy
  protected attributes}. In \bibinfo{booktitle}{\emph{Proceedings of the 2021
  ACM Conference on Fairness, Accountability, and Transparency}}.
  \bibinfo{pages}{237--248}.
\newblock


\bibitem[Mehrotra et~al\mbox{.}(2022)]%
        {mehrotra2022revisiting}
\bibfield{author}{\bibinfo{person}{Anay Mehrotra}, \bibinfo{person}{Jeff
  Sachs}, {and} \bibinfo{person}{L~Elisa Celis}.}
  \bibinfo{year}{2022}\natexlab{}.
\newblock \showarticletitle{Revisiting Group Fairness Metrics: The Effect of
  Networks}.
\newblock \bibinfo{journal}{\emph{Proceedings of the ACM on Human-Computer
  Interaction}} \bibinfo{volume}{6}, \bibinfo{number}{CSCW2}
  (\bibinfo{year}{2022}), \bibinfo{pages}{1--29}.
\newblock


\bibitem[Milkman et~al\mbox{.}(2012)]%
        {milkman2012temporal}
\bibfield{author}{\bibinfo{person}{Katherine~L Milkman},
  \bibinfo{person}{Modupe Akinola}, {and} \bibinfo{person}{Dolly Chugh}.}
  \bibinfo{year}{2012}\natexlab{}.
\newblock \showarticletitle{Temporal distance and discrimination: An audit
  study in academia}.
\newblock \bibinfo{journal}{\emph{Psychological science}} \bibinfo{volume}{23},
  \bibinfo{number}{7} (\bibinfo{year}{2012}), \bibinfo{pages}{710--717}.
\newblock


\bibitem[Miranda et~al\mbox{.}(2023)]%
        {bogen2023towards}
\bibfield{author}{\bibinfo{person}{Bogen Miranda}, \bibinfo{person}{Tripathi
  Pushkar}, \bibinfo{person}{Aditya~Srinivas Timmaraju},
  \bibinfo{person}{Mashayekhi Mehdi}, \bibinfo{person}{Zeng Qi},
  \bibinfo{person}{Roudani Rabyd}, \bibinfo{person}{Gahagan Sean},
  \bibinfo{person}{Howard Andrew}, {and} \bibinfo{person}{Leone Isabella}.}
  \bibinfo{year}{2023}\natexlab{}.
\newblock \bibinfo{booktitle}{\emph{Toward fairness in personalized ads}}.
\newblock \bibinfo{type}{{T}echnical {R}eport}. \bibinfo{institution}{Meta}.
\newblock


\bibitem[Morik et~al\mbox{.}(2020)]%
        {morik2020controlling}
\bibfield{author}{\bibinfo{person}{Marco Morik}, \bibinfo{person}{Ashudeep
  Singh}, \bibinfo{person}{Jessica Hong}, {and} \bibinfo{person}{Thorsten
  Joachims}.} \bibinfo{year}{2020}\natexlab{}.
\newblock \showarticletitle{Controlling fairness and bias in dynamic
  learning-to-rank}. In \bibinfo{booktitle}{\emph{Proceedings of the 43rd
  international ACM SIGIR conference on research and development in information
  retrieval}}. \bibinfo{pages}{429--438}.
\newblock


\bibitem[Newman(2002)]%
        {newman2002assortative}
\bibfield{author}{\bibinfo{person}{M.~E.~J. Newman}.}
  \bibinfo{year}{2002}\natexlab{}.
\newblock \showarticletitle{Assortative Mixing in Networks}.
\newblock \bibinfo{journal}{\emph{Phys. Rev. Lett.}}  \bibinfo{volume}{89}
  (\bibinfo{date}{Oct} \bibinfo{year}{2002}), \bibinfo{pages}{208701}.
\newblock
Issue 20.
\urldef\tempurl%
\url{https://doi.org/10.1103/PhysRevLett.89.208701}
\showDOI{\tempurl}


\bibitem[Newman(2003)]%
        {newman2003mixing}
\bibfield{author}{\bibinfo{person}{M.~E.~J. Newman}.}
  \bibinfo{year}{2003}\natexlab{}.
\newblock \showarticletitle{Mixing patterns in networks}.
\newblock \bibinfo{journal}{\emph{Phys. Rev. E}}  \bibinfo{volume}{67}
  (\bibinfo{date}{Feb} \bibinfo{year}{2003}), \bibinfo{pages}{026126}.
\newblock
Issue 2.
\urldef\tempurl%
\url{https://doi.org/10.1103/PhysRevE.67.026126}
\showDOI{\tempurl}


\bibitem[Newman(2006)]%
        {newman2006finding}
\bibfield{author}{\bibinfo{person}{M.~E.~J. Newman}.}
  \bibinfo{year}{2006}\natexlab{}.
\newblock \showarticletitle{Finding community structure in networks using the
  eigenvectors of matrices}.
\newblock \bibinfo{journal}{\emph{Phys. Rev. E}}  \bibinfo{volume}{74}
  (\bibinfo{date}{Sep} \bibinfo{year}{2006}), \bibinfo{pages}{036104}.
\newblock
Issue 3.
\urldef\tempurl%
\url{https://doi.org/10.1103/PhysRevE.74.036104}
\showDOI{\tempurl}


\bibitem[Nickel et~al\mbox{.}(2011)]%
        {nickel2011three}
\bibfield{author}{\bibinfo{person}{Maximilian Nickel}, \bibinfo{person}{Volker
  Tresp}, {and} \bibinfo{person}{Hans-Peter Kriegel}.}
  \bibinfo{year}{2011}\natexlab{}.
\newblock \showarticletitle{A Three-Way Model for Collective Learning on
  Multi-Relational Data}. In \bibinfo{booktitle}{\emph{Proceedings of the 28th
  International Conference on International Conference on Machine Learning}}
  (Bellevue, Washington, USA) \emph{(\bibinfo{series}{ICML'11})}.
  \bibinfo{publisher}{Omnipress}, \bibinfo{address}{Madison, WI, USA},
  \bibinfo{pages}{809–816}.
\newblock
\showISBNx{9781450306195}


\bibitem[Nilizadeh et~al\mbox{.}(2014)]%
        {nilizadeh2014community}
\bibfield{author}{\bibinfo{person}{Shirin Nilizadeh}, \bibinfo{person}{Apu
  Kapadia}, {and} \bibinfo{person}{Yong-Yeol Ahn}.}
  \bibinfo{year}{2014}\natexlab{}.
\newblock \showarticletitle{Community-enhanced de-anonymization of online
  social networks}. In \bibinfo{booktitle}{\emph{Proceedings of the 2014 acm
  sigsac conference on computer and communications security}}.
  \bibinfo{pages}{537--548}.
\newblock


\bibitem[Olteanu et~al\mbox{.}(2019)]%
        {olteanu2019social}
\bibfield{author}{\bibinfo{person}{Alexandra Olteanu}, \bibinfo{person}{Carlos
  Castillo}, \bibinfo{person}{Fernando Diaz}, {and} \bibinfo{person}{Emre
  K{\i}c{\i}man}.} \bibinfo{year}{2019}\natexlab{}.
\newblock \showarticletitle{Social data: Biases, methodological pitfalls, and
  ethical boundaries}.
\newblock \bibinfo{journal}{\emph{Frontiers in Big Data}}  \bibinfo{volume}{2}
  (\bibinfo{year}{2019}), \bibinfo{pages}{13}.
\newblock


\bibitem[Palla et~al\mbox{.}(2005)]%
        {palla2005uncovering}
\bibfield{author}{\bibinfo{person}{Gergely Palla}, \bibinfo{person}{Imre
  Der{\'e}nyi}, \bibinfo{person}{Ill{\'e}s Farkas}, {and}
  \bibinfo{person}{Tam{\'a}s Vicsek}.} \bibinfo{year}{2005}\natexlab{}.
\newblock \showarticletitle{Uncovering the overlapping community structure of
  complex networks in nature and society}.
\newblock \bibinfo{journal}{\emph{nature}} \bibinfo{volume}{435},
  \bibinfo{number}{7043} (\bibinfo{year}{2005}), \bibinfo{pages}{814--818}.
\newblock


\bibitem[Patro et~al\mbox{.}(2020)]%
        {patro2020fairrec}
\bibfield{author}{\bibinfo{person}{Gourab~K Patro}, \bibinfo{person}{Arpita
  Biswas}, \bibinfo{person}{Niloy Ganguly}, \bibinfo{person}{Krishna~P
  Gummadi}, {and} \bibinfo{person}{Abhijnan Chakraborty}.}
  \bibinfo{year}{2020}\natexlab{}.
\newblock \showarticletitle{Fairrec: Two-sided fairness for personalized
  recommendations in two-sided platforms}. In
  \bibinfo{booktitle}{\emph{Proceedings of the web conference 2020}}.
  \bibinfo{pages}{1194--1204}.
\newblock


\bibitem[Rahmattalabi et~al\mbox{.}(2021)]%
        {rahmattalabi2021welfare}
\bibfield{author}{\bibinfo{person}{Aida Rahmattalabi}, \bibinfo{person}{Shahin
  Jabbari}, \bibinfo{person}{Himabindu Lakkaraju}, \bibinfo{person}{Phebe
  Vayanos}, \bibinfo{person}{Max Izenberg}, \bibinfo{person}{Ryan Brown},
  \bibinfo{person}{Eric Rice}, {and} \bibinfo{person}{Milind Tambe}.}
  \bibinfo{year}{2021}\natexlab{}.
\newblock \showarticletitle{Fair Influence Maximization: a Welfare Optimization
  Approach}.
\newblock \bibinfo{journal}{\emph{Proceedings of the AAAI Conference on
  Artificial Intelligence}} \bibinfo{volume}{35}, \bibinfo{number}{13}
  (\bibinfo{date}{May} \bibinfo{year}{2021}), \bibinfo{pages}{11630--11638}.
\newblock
\urldef\tempurl%
\url{https://doi.org/10.1609/aaai.v35i13.17383}
\showDOI{\tempurl}


\bibitem[Robertson(1977)]%
        {robertson1977probability}
\bibfield{author}{\bibinfo{person}{Stephen~E Robertson}.}
  \bibinfo{year}{1977}\natexlab{}.
\newblock \showarticletitle{The probability ranking principle in IR}.
\newblock \bibinfo{journal}{\emph{Journal of documentation}}
  (\bibinfo{year}{1977}).
\newblock


\bibitem[Roemer and Trannoy(2015)]%
        {roemer2015equality}
\bibfield{author}{\bibinfo{person}{John~E Roemer} {and} \bibinfo{person}{Alain
  Trannoy}.} \bibinfo{year}{2015}\natexlab{}.
\newblock \showarticletitle{Equality of opportunity}.
\newblock In \bibinfo{booktitle}{\emph{Handbook of income distribution}}.
  Vol.~\bibinfo{volume}{2}. \bibinfo{publisher}{Elsevier},
  \bibinfo{pages}{217--300}.
\newblock


\bibitem[Rohe et~al\mbox{.}(2011)]%
        {sbm-recovery}
\bibfield{author}{\bibinfo{person}{Karl Rohe}, \bibinfo{person}{Sourav
  Chatterjee}, {and} \bibinfo{person}{Bin Yu}.}
  \bibinfo{year}{2011}\natexlab{}.
\newblock \showarticletitle{{Spectral clustering and the high-dimensional
  stochastic blockmodel}}.
\newblock \bibinfo{journal}{\emph{The Annals of Statistics}}
  \bibinfo{volume}{39}, \bibinfo{number}{4} (\bibinfo{year}{2011}),
  \bibinfo{pages}{1878 -- 1915}.
\newblock
\urldef\tempurl%
\url{https://doi.org/10.1214/11-AOS887}
\showDOI{\tempurl}


\bibitem[Ronhovde and Nussinov(2009)]%
        {ronhovde2009multiresolution}
\bibfield{author}{\bibinfo{person}{Peter Ronhovde} {and} \bibinfo{person}{Zohar
  Nussinov}.} \bibinfo{year}{2009}\natexlab{}.
\newblock \showarticletitle{Multiresolution community detection for megascale
  networks by information-based replica correlations}.
\newblock \bibinfo{journal}{\emph{Physical Review E}} \bibinfo{volume}{80},
  \bibinfo{number}{1} (\bibinfo{year}{2009}), \bibinfo{pages}{016109}.
\newblock


\bibitem[Rozemberczki and Sarkar(2020)]%
        {feather}
\bibfield{author}{\bibinfo{person}{Benedek Rozemberczki} {and}
  \bibinfo{person}{Rik Sarkar}.} \bibinfo{year}{2020}\natexlab{}.
\newblock \showarticletitle{{Characteristic Functions on Graphs: Birds of a
  Feather, from Statistical Descriptors to Parametric Models}}. In
  \bibinfo{booktitle}{\emph{Proceedings of the 29th ACM International
  Conference on Information and Knowledge Management (CIKM '20)}}. ACM,
  \bibinfo{pages}{1325–1334}.
\newblock


\bibitem[Schelling(1971)]%
        {schelling1971}
\bibfield{author}{\bibinfo{person}{Thomas~C. Schelling}.}
  \bibinfo{year}{1971}\natexlab{}.
\newblock \showarticletitle{Dynamic models of segregation}.
\newblock \bibinfo{journal}{\emph{The Journal of Mathematical Sociology}}
  \bibinfo{volume}{1}, \bibinfo{number}{2} (\bibinfo{year}{1971}),
  \bibinfo{pages}{143--186}.
\newblock
\urldef\tempurl%
\url{https://doi.org/10.1080/0022250X.1971.9989794}
\showDOI{\tempurl}


\bibitem[Shorrocks(1980)]%
        {shorrocks1980class}
\bibfield{author}{\bibinfo{person}{Anthony~F Shorrocks}.}
  \bibinfo{year}{1980}\natexlab{}.
\newblock \showarticletitle{The class of additively decomposable inequality
  measures}.
\newblock \bibinfo{journal}{\emph{Econometrica: Journal of the Econometric
  Society}} (\bibinfo{year}{1980}), \bibinfo{pages}{613--625}.
\newblock


\bibitem[Shorrocks(1984)]%
        {shorrocks1984inequality}
\bibfield{author}{\bibinfo{person}{Anthony~F Shorrocks}.}
  \bibinfo{year}{1984}\natexlab{}.
\newblock \showarticletitle{Inequality decomposition by population subgroups}.
\newblock \bibinfo{journal}{\emph{Econometrica: Journal of the Econometric
  Society}} (\bibinfo{year}{1984}), \bibinfo{pages}{1369--1385}.
\newblock


\bibitem[Singh and Joachims(2018)]%
        {singh2018fairness}
\bibfield{author}{\bibinfo{person}{Ashudeep Singh} {and}
  \bibinfo{person}{Thorsten Joachims}.} \bibinfo{year}{2018}\natexlab{}.
\newblock \showarticletitle{Fairness of exposure in rankings}. In
  \bibinfo{booktitle}{\emph{Proceedings of the 24th ACM SIGKDD International
  Conference on Knowledge Discovery \& Data Mining}}.
  \bibinfo{pages}{2219--2228}.
\newblock


\bibitem[Speicher et~al\mbox{.}(2018)]%
        {speicher2018unified}
\bibfield{author}{\bibinfo{person}{Till Speicher}, \bibinfo{person}{Hoda
  Heidari}, \bibinfo{person}{Nina Grgic-Hlaca}, \bibinfo{person}{Krishna~P
  Gummadi}, \bibinfo{person}{Adish Singla}, \bibinfo{person}{Adrian Weller},
  {and} \bibinfo{person}{Muhammad~Bilal Zafar}.}
  \bibinfo{year}{2018}\natexlab{}.
\newblock \showarticletitle{A unified approach to quantifying algorithmic
  unfairness: Measuring individual \&group unfairness via inequality indices}.
  In \bibinfo{booktitle}{\emph{Proceedings of the 24th ACM SIGKDD international
  conference on knowledge discovery \& data mining}}.
  \bibinfo{pages}{2239--2248}.
\newblock


\bibitem[Stoica et~al\mbox{.}(2020)]%
        {stoica2020seeding}
\bibfield{author}{\bibinfo{person}{Ana-Andreea Stoica},
  \bibinfo{person}{Jessy~Xinyi Han}, {and} \bibinfo{person}{Augustin
  Chaintreau}.} \bibinfo{year}{2020}\natexlab{}.
\newblock \showarticletitle{Seeding Network Influence in Biased Networks and
  the Benefits of Diversity}. In \bibinfo{booktitle}{\emph{Proceedings of The
  Web Conference 2020}} \emph{(\bibinfo{series}{WWW '20})}.
  \bibinfo{publisher}{Association for Computing Machinery},
  \bibinfo{pages}{2089–2098}.
\newblock
\showISBNx{9781450370233}
\urldef\tempurl%
\url{https://doi.org/10.1145/3366423.3380275}
\showDOI{\tempurl}


\bibitem[Sweeney(2013)]%
        {sweeney2013discrimination}
\bibfield{author}{\bibinfo{person}{Latanya Sweeney}.}
  \bibinfo{year}{2013}\natexlab{}.
\newblock \showarticletitle{Discrimination in online ad delivery}.
\newblock \bibinfo{journal}{\emph{Queue}} \bibinfo{volume}{11},
  \bibinfo{number}{3} (\bibinfo{year}{2013}), \bibinfo{pages}{10}.
\newblock


\bibitem[Thelwall(2009)]%
        {thelwall2009homophily}
\bibfield{author}{\bibinfo{person}{Mike Thelwall}.}
  \bibinfo{year}{2009}\natexlab{}.
\newblock \showarticletitle{Homophily in myspace}.
\newblock \bibinfo{journal}{\emph{Journal of the American Society for
  Information Science and Technology}} \bibinfo{volume}{60},
  \bibinfo{number}{2} (\bibinfo{year}{2009}), \bibinfo{pages}{219--231}.
\newblock


\bibitem[Tomasev et~al\mbox{.}(2021)]%
        {tomasev2021fairness}
\bibfield{author}{\bibinfo{person}{Nenad Tomasev}, \bibinfo{person}{Kevin~R
  McKee}, \bibinfo{person}{Jackie Kay}, {and} \bibinfo{person}{Shakir
  Mohamed}.} \bibinfo{year}{2021}\natexlab{}.
\newblock \showarticletitle{Fairness for unobserved characteristics: Insights
  from technological impacts on queer communities}. In
  \bibinfo{booktitle}{\emph{Proceedings of the 2021 AAAI/ACM Conference on AI,
  Ethics, and Society}}. \bibinfo{pages}{254--265}.
\newblock


\bibitem[Tsang et~al\mbox{.}(2019)]%
        {tsang2019group}
\bibfield{author}{\bibinfo{person}{Alan Tsang}, \bibinfo{person}{Bryan Wilder},
  \bibinfo{person}{Eric Rice}, \bibinfo{person}{Milind Tambe}, {and}
  \bibinfo{person}{Yair Zick}.} \bibinfo{year}{2019}\natexlab{}.
\newblock \showarticletitle{Group-Fairness in Influence Maximization}. In
  \bibinfo{booktitle}{\emph{Proceedings of the Twenty-Eighth International
  Joint Conference on Artificial Intelligence, {IJCAI-19}}}.
  \bibinfo{publisher}{International Joint Conferences on Artificial
  Intelligence Organization}, \bibinfo{pages}{5997--6005}.
\newblock
\urldef\tempurl%
\url{https://doi.org/10.24963/ijcai.2019/831}
\showDOI{\tempurl}


\bibitem[Ugander et~al\mbox{.}(2011)]%
        {ugander2011anatomy}
\bibfield{author}{\bibinfo{person}{Johan Ugander}, \bibinfo{person}{Brian
  Karrer}, \bibinfo{person}{Lars Backstrom}, {and} \bibinfo{person}{Cameron
  Marlow}.} \bibinfo{year}{2011}\natexlab{}.
\newblock \showarticletitle{The anatomy of the facebook social graph}.
\newblock \bibinfo{journal}{\emph{arXiv preprint arXiv:1111.4503}}
  (\bibinfo{year}{2011}).
\newblock


\bibitem[Usunier et~al\mbox{.}(2022)]%
        {usunier2022fast}
\bibfield{author}{\bibinfo{person}{Nicolas Usunier}, \bibinfo{person}{Virginie
  Do}, {and} \bibinfo{person}{Elvis Dohmatob}.}
  \bibinfo{year}{2022}\natexlab{}.
\newblock \showarticletitle{Fast online ranking with fairness of exposure}. In
  \bibinfo{booktitle}{\emph{2022 ACM Conference on Fairness, Accountability,
  and Transparency}}. \bibinfo{pages}{2157--2167}.
\newblock


\bibitem[Veale and Binns(2017)]%
        {veale2017fairer}
\bibfield{author}{\bibinfo{person}{Michael Veale} {and} \bibinfo{person}{Reuben
  Binns}.} \bibinfo{year}{2017}\natexlab{}.
\newblock \showarticletitle{Fairer machine learning in the real world:
  Mitigating discrimination without collecting sensitive data}.
\newblock \bibinfo{journal}{\emph{Big Data \& Society}} \bibinfo{volume}{4},
  \bibinfo{number}{2} (\bibinfo{year}{2017}),
  \bibinfo{pages}{2053951717743530}.
\newblock


\bibitem[Verbrugge(1977)]%
        {verbrugge1977structure}
\bibfield{author}{\bibinfo{person}{Lois~M. Verbrugge}.}
  \bibinfo{year}{1977}\natexlab{}.
\newblock \showarticletitle{The Structure of Adult Friendship Choices}.
\newblock \bibinfo{journal}{\emph{Social Forces}} \bibinfo{volume}{56},
  \bibinfo{number}{2} (\bibinfo{year}{1977}), \bibinfo{pages}{576--597}.
\newblock
\showISSN{00377732, 15347605}
\urldef\tempurl%
\url{http://www.jstor.org/stable/2577741}
\showURL{%
\tempurl}


\bibitem[Wachter and Mittelstadt(2019)]%
        {wachter2019right}
\bibfield{author}{\bibinfo{person}{Sandra Wachter} {and} \bibinfo{person}{Brent
  Mittelstadt}.} \bibinfo{year}{2019}\natexlab{}.
\newblock \showarticletitle{A right to reasonable inferences: re-thinking data
  protection law in the age of big data and AI}.
\newblock \bibinfo{journal}{\emph{Colum. Bus. L. Rev.}} (\bibinfo{year}{2019}),
  \bibinfo{pages}{494}.
\newblock


\bibitem[Wang and Joachims(2021)]%
        {wang2021user}
\bibfield{author}{\bibinfo{person}{Lequn Wang} {and} \bibinfo{person}{Thorsten
  Joachims}.} \bibinfo{year}{2021}\natexlab{}.
\newblock \showarticletitle{User fairness, item fairness, and diversity for
  rankings in two-sided markets}. In \bibinfo{booktitle}{\emph{Proceedings of
  the 2021 ACM SIGIR International Conference on Theory of Information
  Retrieval}}. \bibinfo{pages}{23--41}.
\newblock


\bibitem[Wright et~al\mbox{.}(2013)]%
        {wright2013religious}
\bibfield{author}{\bibinfo{person}{Bradley~RE Wright}, \bibinfo{person}{Michael
  Wallace}, \bibinfo{person}{John Bailey}, {and} \bibinfo{person}{Allen Hyde}.}
  \bibinfo{year}{2013}\natexlab{}.
\newblock \showarticletitle{Religious affiliation and hiring discrimination in
  New England: A field experiment}.
\newblock \bibinfo{journal}{\emph{Research in Social Stratification and
  Mobility}}  \bibinfo{volume}{34} (\bibinfo{year}{2013}),
  \bibinfo{pages}{111--126}.
\newblock


\bibitem[Wright(1997)]%
        {wright1997class}
\bibfield{author}{\bibinfo{person}{Erik~Olin Wright}.}
  \bibinfo{year}{1997}\natexlab{}.
\newblock \bibinfo{booktitle}{\emph{Class counts: Comparative studies in class
  analysis}}.
\newblock \bibinfo{publisher}{Cambridge University Press}.
\newblock


\bibitem[Yin et~al\mbox{.}(2017)]%
        {email}
\bibfield{author}{\bibinfo{person}{Hao Yin}, \bibinfo{person}{Austin~R.
  Benson}, \bibinfo{person}{Jure Leskovec}, {and} \bibinfo{person}{David~F.
  Gleich}.} \bibinfo{year}{2017}\natexlab{}.
\newblock \showarticletitle{Local Higher-Order Graph Clustering}. In
  \bibinfo{booktitle}{\emph{Proceedings of the 23rd ACM SIGKDD International
  Conference on Knowledge Discovery and Data Mining}}
  \emph{(\bibinfo{series}{KDD '17})}. \bibinfo{publisher}{Association for
  Computing Machinery}, \bibinfo{address}{New York, NY, USA},
  \bibinfo{pages}{555–564}.
\newblock
\urldef\tempurl%
\url{https://doi.org/10.1145/3097983.3098069}
\showDOI{\tempurl}


\bibitem[Zhang(2018)]%
        {zhang2018assessing}
\bibfield{author}{\bibinfo{person}{Yan Zhang}.}
  \bibinfo{year}{2018}\natexlab{}.
\newblock \showarticletitle{Assessing fair lending risks using race/ethnicity
  proxies}.
\newblock \bibinfo{journal}{\emph{Management Science}} \bibinfo{volume}{64},
  \bibinfo{number}{1} (\bibinfo{year}{2018}), \bibinfo{pages}{178--197}.
\newblock


\bibitem[Zhang and Rohe(2018)]%
        {zhang2018understanding}
\bibfield{author}{\bibinfo{person}{Yilin Zhang} {and} \bibinfo{person}{Karl
  Rohe}.} \bibinfo{year}{2018}\natexlab{}.
\newblock \showarticletitle{Understanding Regularized Spectral Clustering via
  Graph Conductance}. In \bibinfo{booktitle}{\emph{Proceedings of the 32nd
  International Conference on Neural Information Processing Systems}}
  (Montr\'{e}al, Canada) \emph{(\bibinfo{series}{NIPS'18})}.
  \bibinfo{publisher}{Curran Associates Inc.}, \bibinfo{address}{Red Hook, NY,
  USA}, \bibinfo{pages}{10654–10663}.
\newblock


\bibitem[Zheleva et~al\mbox{.}(2012)]%
        {zheleva2012privacy}
\bibfield{author}{\bibinfo{person}{Elena Zheleva}, \bibinfo{person}{Evimaria
  Terzi}, {and} \bibinfo{person}{Lise Getoor}.}
  \bibinfo{year}{2012}\natexlab{}.
\newblock \showarticletitle{Privacy in social networks}.
\newblock \bibinfo{journal}{\emph{Synthesis Lectures on Data Mining and
  Knowledge Discovery}} \bibinfo{volume}{3}, \bibinfo{number}{1}
  (\bibinfo{year}{2012}), \bibinfo{pages}{1--85}.
\newblock


\end{thebibliography}

\appendix
\section{Proofs}
\subsection{Proposition \ref{prop:additive-decomposability}} \label{sec:decomposability}
\begin{proof}
First, we show that Equation \autoref{eqn:prop-decomposability} implies additive decomposability. This can be shown by defining the kernel given a partition $\mathcal{G}$ of $\{1,\ldots,n\}$; namely define the kernel to be the ground-truth kernel where $K^*_{ij} = 1 / \lvert g \lvert$ if $i$ and $j$ are both in group $g \in \mathcal{G}$, otherwise $K^*_{ij} = 0$. Plugging this kernel into Equation \autoref{eqn:prop-decomposability}, we have exactly the condition for additive decomposability in Equation \autoref{ref:decomposability}.

Second, first notice that if columns of $K$ do not sum to 1 but still sum to the same value $c>0$, then invariances to the scale of $A$ and $F$ mean that the same decomposition holds by replacing $n$ with $cn$ in the denominator. We prove this more general equality.

We show that additive decomposability implies Equation \autoref{eqn:prop-decomposability}. Let us begin by considering $K$ such that all entries are integer. In this case, we can define partitions $\mathcal{G}$ given the kernel. Think of the constant $C$, the sum of each column, as the number of times each element in $x$, as defined in Equation \autoref{eq:weighted_entropy}, should be repeated. Repeat $x$ and $w$ $C$ times, which does not affect the inequality thanks to the property of replication invariance $F\left(x^{\otimes C}, w^{\otimes C}\right) = F(x, w)$. Then construct $n$ groups, one for each $i$, with $K_{ij}$ repetitions of $x_j$. If we apply the condition of additive decomposability in Equation~\autoref{ref:decomposability} to these partitions, we recover Equation~\autoref{eqn:prop-decomposability}. Now, we extend to kernels with rational values, let $c$ be the smallest common multiple of the denominator of values in $K_{ij}$ and denote by $\tilde{K} = c K$. The invariance to input scale $K$ by applying the properties that $F$ is continuous and invariant to input scaling, which allows for extending from integers to rationals. We can then extend to the whole set of real values by continuity.
\footnote{The extension from rationals to real values requires uniform continuity. (Weighted) generalized entropies have singularities at $0$, when the sum of the weights $K_i^\top\Vec{1}=0$ or the mean $A(K, \datavec)_i=0$. However, we only need to extend by continuity to non-zero entries of $K$ since $0$ is rational, so let us define $k_0 = 0.5\min_{i,j:K_{ij}>0} K_{ij}$ and let $\tilde{\mathcal{K}} = \Big\{\tilde{K}\in\Re^{n \times n}: \forall i,j, \text{either }K_{ij}=0 \text{ and } \tilde{K}_{ij}=0 \text{ or } \tilde{K}_{ij}>k_0 \Big\}$. Then it is easy to check that $\tilde{K}\mapsto q(A(\tilde{K}, \datavec))$ and $\tilde{K}\mapsto F(A(\tilde{K}, \datavec), \tilde{K}\Vec{1})$ are uniformly continuous on $\tilde{\mathcal{K}}$.}
\end{proof}

\subsection{Lemmas for Proposition \ref{prop:inequality-bounds}}\label{sec:lemmas-bounds}
We begin by considering a setting where nodes in separate sensitive groups share some similarity. Figure \ref{fig:confounder-blend} illustrates such a setting. The left kernel reflects the sensitive groups whereas the right kernel introduces between-group similarities between sensitive groups. The parameters $p, q$ are the values of the kernel matrix where $p > q$.
\begin{figure}
    \centering
    \includegraphics[width=0.5\linewidth]{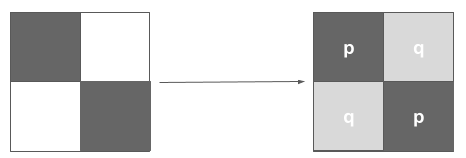}
    \caption{The kernel on the left reflects the groups separated by a sensitive attribute, where for the purpose of example the two classes are of equal size. Next, we introduce a level of similarity between nodes in differing sensitive groups yielding the kernel on the right, where values of the kernel matrix are shown in white and $q < p$. The addition of $q$ dampens the between-group inequality.}
    \label{fig:confounder-blend}
\end{figure}
\begin{lemma}\label{lemma:blend}
    Let a protected attribute partition the population into sensitive groups of equal size. Further, suppose two nodes in the same sensitive class have similarity $p$ whereas two nodes in differing classes have similarity $q$ where $q < p$. We call the similarity between nodes in separate groups the ``blending" similarity. Then, with the blending similarity, between-group inequality decreases according to:
    \begin{equation*}
        \Delta'_b = \left(\frac{p - q}{p + q}\right)^2 \Delta^0_b 
    \end{equation*}
    Where $\Delta^0_b$ is the between-group inequality without the blending similarity (left kernel in Figure \ref{fig:confounder-blend}) and $\Delta'$ is the between-group inequality with the blending similarity.
\end{lemma}

\begin{proof}
    Let there be a population $V = \{1, 2, \cdots, 2n\}$ and a sensitive attribute that partitions the population into two halves $V_1 = \{1, 2, \cdots, n\}, V_2 = \{n+1, n+2, \cdots, 2n\}$. In the case without the between-group similarity, let us define the ground-truth kernel matrix $K_s\in \mathbb{R}^{2n \times 2n}$, where $K_s(ij)$, the similarity between nodes $i$ and $j$, is:
    \begin{equation*}
    K_s(ij) = \begin{cases} 
          1 & i, j \in V_1 \cup i, j \in V_2\\
          0 & \text{otherwise} \\
       \end{cases}
    \end{equation*}
    Let the  variable $\datavec_i\in [0, 1]$ denote the allocation of a desirable quantity to individual $i$. Further let $\mu_1 = \frac{1
    }{n} \sum_{i \in V_1} \datavec_i$ be the average allocation for $V_1$ ($\mu_2$ is defined analogously). Let the inequality function $F$ be the normalized variance.

    Now let us examine the between-group inequality with the blending similarity. Consistent with Figure \ref{fig:confounder-blend}, let us define the kernel similarity values between two nodes $i, j$ as:
    \begin{equation*}
    K_{ij} = \begin{cases} 
          p & i, j \in V_1 \cup i, j \in V_2\\
          q & \text{otherwise} 
       \end{cases}
    \end{equation*}
    To achieve additive decomposability, the columns of the kernel matrix are normalized. Now, applying Equation \ref{eq:weighted-neigh}, the between-group inequality with the blending similarity is:
    \begin{align}
        \Delta'_b &= F(A(K, \datavec))\\
        &= F\left(\frac{p\mu_1 + q\mu_2}{p + q}, \cdots, \frac{p\mu_1 + q\mu_2}{p + q}, \frac{q\mu_1 + p\mu_2}{p + q}, \cdots, \frac{q\mu_1 + p\mu_2}{p + q}\right)\\
        &= F\left(p\mu_1 + q\mu_2, \cdots, p\mu_1 + q\mu_2, q\mu_1 + p\mu_2, \cdots, q\mu_1 + p\mu_2\right) \label{eqn:blend_normvar}
    \end{align}
    In the final line, we use the fact that $F(cX) = F(X)$. Now, we re-write Equation \ref{eqn:blend_normvar} as a change of variables for Equation \ref{eqn:ground-truth-inequality} to get:
    \begin{align}
        \Delta'_b &= \left[\left(\frac{p-q}{p+q}\right)\left(\frac{\mu_1 - \mu_2}{\mu_1 + \mu_2}\right)\right]^2\\
        &= \left(\frac{p-q}{p+q}\right)^2 \Delta^0_b
    \end{align}
\end{proof}

Both lemmas below utilize the Transfer Principle of inequality functions, which follows from Schur convexity. Let $\datavec$ be a vector of node outcomes, and let $\datavec_i$ be the $i$th largest value in $\datavec$. Then, the Transfer Principle states that for all $1 \leq i < j \leq n$ and $\delta < \left(\datavec_k - \datavec_{k+1}\right)/2 \forall k \in \{1, \cdots, n-1\}$:
\begin{equation*}
    F\left(\cdots, \datavec_i - \delta, \cdots , \datavec_j + \delta, \cdots\right) < F(\datavec)
\end{equation*}
Intuitively, the Transfer Principle states that if the same amount is taken from a better-off individual and given to a worse-off individual, inequality must decrease. 

\begin{lemma}\label{lemma:maxmin}
    Let $\datavec \in \{0, 1\}^n$ be a vector of outcomes such that $\sum_{i=1}^n \datavec_i = n/2$ and $K_s$ be the ground-truth kernel corresponding to a sensitive attribute that partitions the population in half. Then, let $\datavec_s$ be the ground-truth group averages: $\datavec_s = A\left(K_s, \datavec\right)$. 

    We now consider the set of all confounding kernels $K_C = \sum K_c$ such that each $K_c$ is row and column normalized and the node permutation of a block-diagonal matrix, with blocks of equal size. Each confounding kernel yields a smoothed vector $\datavec_C = A\left(K_C, \datavec\right)$. The $\datavec_C$ that maximizes group-free inequality, where $F$ is the normalized variance and $\alpha \in [0, 1]$, is:
    \begin{equation*}
        \datavec = \argmax_{\datavec_C} F\left(\alpha \datavec_s + (1 - \alpha) \datavec_C\right)
    \end{equation*}
    The smoothed vector $\datavec_C = A\left(K_C, \datavec\right)$ that minimizes group-free inequality is:
    \begin{equation*}
        \overline{\mathbf{y}} = \frac{\mathbf{y}^T\Vec{1}}{n}\Vec{1} = \argmin_{\datavec_C} F\left(\alpha \datavec_s + (1 - \alpha) \datavec_C\right)
    \end{equation*}
\end{lemma}
\begin{proof}
    \textbf{Maximizer} First, let us observe that $\datavec_C = \datavec$ is the only smoothed confounding vector where all entries are integer, given that a node will always have similarity with itself. Thus, any other smoothed vector $\datavec_C'$ must have fractional smoothed values. We will now show that such a vector $\datavec_C'$ can be modified to yield a higher inequality and is thus not the maximizing confounding vector. 

    Let node $i$ be a node with fractional smoothed value $\datavec_C'(i) \in (0, 1)$. It must then be the case that there is a confounding attribute $c$ that partitions the population such that at least two partitions contain a mixture of $0$ and $1$. We show that it cannot be the case that there is only one partition with mixed values. Assume that the single fractional value equals $\max A\left(K_c, \datavec\right)$, then all other confounding partitions must have a smoothed average of $0$. In this case, the partition containing node $i$ must have a smoothed average of $1$, as all nodes with value $1$ must be in the same partition as node $i$ and there can be no zeros in the partition for partitions to be of equal size (unless the entire graph is one group, which corresponds to the minimizer. Thus, in the case where $\datavec_c'(i) = \max \datavec_c'$, it must be that $\datavec_c' = \datavec$, contradicting the definition of $\datavec_c'$ as a fractional smoothed vector. For an analogous reason, there cannot be a single fractional value where all other partitions have an average of $1$.

    Let node $i$ belong to the fraction partition with the smaller smoothed average in $A\left(K, \datavec\right) = \alpha\datavec_s + (1 - \alpha)\datavec_C$. It is possible for the partition containing node $i$ to swap a node of value zero with a node of value $1$ in the partition with a higher smoothed average. By the Transfer Principle, this swap increases the inequality thus $\datavec_C'$ is not the maximizer.

    \textbf{Minimizer} Let $\datavec_0$ be the smoothed outcomes vector when $\datavec_C = \overline{\datavec}$: $\datavec_0 = \alpha\datavec_s + (1 - \alpha)\overline{\datavec}$. Recall that due to the sensitive attributes construction, the components of $\datavec_0$ take on two values: $\{\alpha\highavg + \frac{1 - \alpha}{2}, \alpha\lowavg + \frac{1 - \alpha}{2}\}$. Let us define these as $\datavec_0^+$ and $\datavec_0^-$ respectively.

    Now, we prove by contradiction. Consider an alternative smoothed vector $\datavec' = \alpha\datavec_s + (1 - \alpha)\datavec_C$ where $\datavec_C \ne \overline{\datavec}$ and $F\left(\datavec'\right) < F\left(\datavec_0\right)$. Now it must be the case that there exists a node $i$ where $\datavec_0(i) = \datavec_0^+$ and $\datavec'(i) > \datavec_0^+$. Further there must exist a corresponding node $j$ where $\datavec_0(j) = \datavec_0^-$ and $\datavec'(j) < \datavec_0^-$. It is then possible for a confounding partition containing node $i$ to swap a node with the confounding partition containing node $j$. In doing so, the Transfer Principle implies the inequality must decrease, suggesting that any $\datavec_C \ne \overline{\datavec}$ is not the minimizer. Note that after the swap, $\datavec'(j)$ remains less than $\datavec'(i)$ due to the outcome imbalance between the sensitive groups.
\end{proof}

\subsection{Proposition \ref{prop:inequality-bounds}} \label{sec:proof-bounds}
\begin{proof}
The structure of the proof is as follows: we first characterize the confounding kernels $\sum_C K_c$ corresponding to the upper and lower bounds. Then, we insert these confounding kernels into our measure of group-free inequality to derive the closed-form bounds. 

Because each confounding attribute divides the population into partitions of equal size and weight ($K_c\Vec{1} = \lambda_c\Vec{1}$) we can simplify the group-free inequality measure. The following simplifications are due to the scale-invariant property of the inequality function:
\begin{align}
    \Delta_b' &= F\left(A(K, \datavec), K\Vec{1}\right)\\
    &= F\left(A(K_s + K_C, \datavec), \left(K_s + K_C\right)\Vec{1}\right)\\
    &= F\left(A(K_s + K_C, \datavec), \left(\lambda_c + \frac{np}{2}\right)\Vec{1}\right)\\
    &= F\left(A(K_s + K_C, \datavec)\right)
\end{align}

Depending on the values of $p$ and $q$, the value of $A(K_s + K_C, \datavec)$ can be re-parameterized as the weighted sum of $A\left(K_s, \datavec\right)$ and $A\left(K_C, \datavec\right)$. In the below, we set $\lvert K_C \lvert_1 = qn^2/2$: 
\begin{align}
    A\left(K_s + K_C, \datavec\right) &= \frac{K_s\datavec + K_C\datavec}{\left(K_s + K_C\right)\Vec{1}}\\
    &= \frac{\left(K_s\Vec{1}\right) \otimes A\left(K_s, \datavec\right) + \left(K_C\Vec{1}\right) \otimes A\left(K_C, \datavec\right)}{\left(K_s + K_C\right)\Vec{1}}\\
    &= \frac{\left(\frac{n}{2}p\right) A\left(K_s, \datavec\right) + \left(\frac{n}{2}q\right) A\left(K_C, \datavec\right)}{\frac{n}{2}(p+q)}\\
    &= \frac{p}{p+q} A\left(K_s, \datavec\right) + \frac{q}{p+q} A\left(K_C, \datavec\right)\\
\end{align}
Let us re-parameterize the weighted sum with parameter $\alpha = p / (p+q)$ that interpolates between $A\left(K_s, \datavec\right)$ and $A\left(K_C, \datavec\right)$. Let $\datavec_s = A\left(K_s, \datavec\right)$ be the fixed, ground-truth group averages whereas let $\datavec_C= A\left(K_C, \datavec\right)$ be variable (depending on the kernel $K_C$) and the average of $\datavec$ based on the confounding attributes. We then have:
\begin{align}
    \Delta_b' &= F\left(A(K_s + K_C, \datavec)\right)\\
    &= F\left(\alpha A\left(K_s, \datavec\right) + (1 - \alpha) A\left(K_C, \datavec\right)\right)\\
    &= F\left(\alpha \datavec_s + (1 - \alpha) \datavec_C\right) \label{eqn:confounding-inequality}
\end{align}

Bounding our group-free inequality measure in the presence of confounding attributes then boils down to bounding the inequality measure characterized in Equation \ref{eqn:confounding-inequality}. In Lemma \ref{lemma:maxmin}, we establish the values of $\datavec_C$ that correspond with the minimum and maximum values of group-free inequality. 

Given that $F$ is additively decomposable, we know that the minimum value of group-free inequality is zero and corresponds with $A\left(K, \datavec\right) = \overline{\mathbf{y}}$; in other words, when the kernel represents the entire population as one group. On the other hand, the maximum value corresponds with $A\left(K, \datavec\right) = \datavec$, when two nodes share similarity only with other nodes of the same label. Lemma \ref{lemma:maxmin} shows that linearly combining an initial, fixed smoothed vector $\datavec_s$ with $\datavec$ increases inequality the most and combining $\datavec_s$ with $\overline{\datavec}$ decreases inequality the most.

The remainder of the proof consists of computing our group-free inequality measure using the confounding kernels corresponding to the upper and lower bounds. We analyze the setting in which $K_C = qn^2/2$, as the bounds widen as the weight of the confounding kernel increases. We also note that in computing closed-form bounds it is sufficient to analyze any $K_C$ that satisfies $A\left(K_C, \datavec\right) = \datavec$ for the upper bound and $A\left(K_C, \datavec\right) = \overline{\datavec}$ for the lower bound due to Equation \ref{eqn:confounding-inequality}.

For the lower bound, we analyze the case in which the confounding attributes create similarities between all pairs of nodes, eroding the group structure from the sensitive attribute. We can see that the kernel $K_C$ that takes on a constant value of $q/2$ for all $i, j$ yields $A\left(K_C, \datavec\right) = \overline{\datavec}$. Then kernel $K = K_s + K_C$ is equivalent in structure to the kernel shown in Figure \ref{fig:confounder-blend}, replaced with $p'$ and $q'$, where $p' = p + \frac{q}{2}$ and $q' = \frac{q}{2}$. Using Lemma \ref{lemma:blend}, the between-group inequality is then:
 \begin{align*}
     \Delta'_b &= \left(\frac{p' - q'}{p' + q'}\right)^2 \Delta^0_b\\
     &= \left(\frac{p + q/2 - q/2}{p + q/2 + q/2}\right)^2 \Delta^0_b\\
     &= \left(\frac{p}{p + q}\right)^2 \Delta_b^0
 \end{align*}

For the upper bound, we consider the case in which the confounding attributes exactly identify the nodes with a positive label $\datavec_i = 1$. In this case, the confounding kernel $K_C = \sum_C K_c$ has the form: 
\begin{equation*}
K_C(i,j) = \begin{cases} 
      q & \datavec_i = \datavec_j = 1\\
      q & \datavec_i = \datavec_j = 0\\
      0 & \text{otherwise} 
   \end{cases}
\end{equation*}
\begin{figure}
    \centering
    \includegraphics[width=0.9\linewidth]{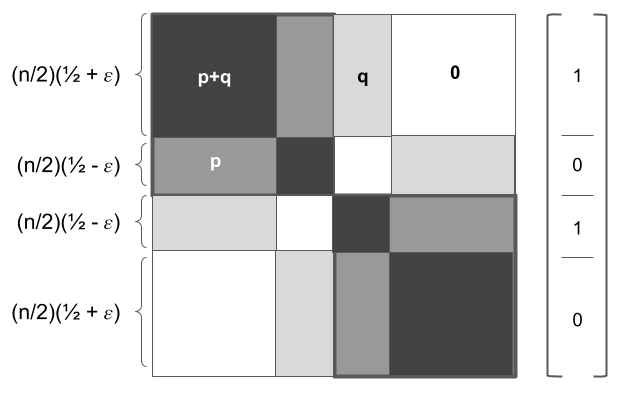}
    \caption{The figure above visualizes the smoothed values $A(K, \datavec) = K\datavec$ for the upper bound in Proposition \ref{prop:inequality-bounds}. There are two sensitive classes, each of size $\frac{n}{2}$; in one class, the ratio of nodes with positive labels is $\frac{1}{2} + \epsilon$ and in the second group the ratio is $\frac{1}{2} - \epsilon$.}
    \label{fig:upper-bound-kernel}
\end{figure}
 The calculation of the smoothed values $A\left(K, \datavec\right)$ is shown in Figure \ref{fig:upper-bound-kernel}. The four values in the product, from top to bottom are:
 \begin{align}
     \frac{(p+q)\highavg + q\lowavg}{(p +q)\highavg + p\lowavg + q\lowavg} &= \frac{p\highavg + q}{p + q}\\
     \frac{p\highavg}{(p+q)\lowavg + p\highavg + q\highavg} &= \frac{p\highavg}{p + q}\\
     \frac{(p+q)\lowavg + q\highavg}{(p+q)\lowavg + p\highavg + q\highavg} &= \frac{p\lowavg + q}{p + q}\\
     \frac{p\lowavg}{(p+q)\highavg + p\lowavg + q\lowavg} &= \frac{p\lowavg}{p + q}
 \end{align}
 Because the normalized variance is invariant to scaling by a constant, we can disregard the $p+q$ denominators in the four above terms. The between-group inequality is then:
 \begin{align*}
     \begin{split}
     F(A(K, \datavec)) = F\biggl(&\left[p\highavg + q, p\highavg, p\lowavg + q, p\lowavg\right],\\
     &\left[\frac{1}{2} + \epsilon, \frac{1}{2} - \epsilon, \frac{1}{2} - \epsilon, \frac{1}{2} + \epsilon\right]\biggr)
     \end{split}
 \end{align*}

To calculate the normalized variance, we first calculate the mean: 
 \begin{align*}
     \mu &= \frac{1}{\sum w_i}\sum w_i \datavec_i\\
     &= \frac{1}{2}\biggl[\highavg\left(p\highavg + q\right) + \lowavg p\highavg \\
     &\quad + \lowavg\left(p\lowavg + q\right) + \highavg\left(p\lowavg + q\right)\biggr]\\
     &= \frac{1}{2}\left(p\left(\highavg^2 + \lowavg^2\right)  + 2p\highavg\lowavg + q\right)\\
     &= \frac{1}{2}\left(p\left(\highavg + \lowavg\right)^2 + q\right)\\
     &= \frac{p + q}{2}
 \end{align*}
 The normalized variance itself is then:
  \begin{align*}
     F(A(K, \datavec)) &= F\biggl(\left[p\highavg + q, p\highavg, p\lowavg + q, p\lowavg\right],\\
     &\quad \left[\frac{1}{2} + \epsilon, \frac{1}{2} - \epsilon, \frac{1}{2} - \epsilon, \frac{1}{2} + \epsilon\right]\biggr)\\
     &= \frac{1}{\mu^2 \sum w_i} \sum w_i\left(x_i - \mu\right)^2\\
     &= \frac{4p\epsilon^2(p + 2q) + q^2}{(p+q)^2}\\
     &= \Delta_b^0 + \left(\frac{q}{p + q}\right)^2\left(1 - \Delta_b^0\right)
 \end{align*}
 \end{proof}

\section{Additional Experimental Details}
\label{sec:app-exps}
\subsection{Network Preprocessing}
We filter the networks to highlight the community structure in each. For \emaileu, \lastfm, and \deezer, we filter for nodes with at least five neighbors, and across all four networks we filter for ground-truth classes with at least $100$ nodes. For \lastfm{} and \deezer{}, we filter for users who have provided at least $50$ ratings on the respective platforms. Finally, we take the largest connected component of the remaining graph.

When computing the kernel with Laplacian Eigenmaps we process all networks as undirected networks.

\subsection{Node classification node feature analysis}\label{sec:node-classification-node-features}
 In Figures \ref{fig:polblogs_feature_activation_correlation} and \ref{fig:email_feature_activation_correlation}, we visualize the correlation between node significance features, such as degree and betweenness centrality, and prevalence of positive labels following post-processing. For \polblogs{}, nodes of lower degree as well as nodes of lower centrality are more likely to be labeled positively following post-processing whereas for \emaileu{}, the trend is generally reversed.
\begin{figure}
     \centering
     \begin{subfigure}[b]{0.32\textwidth}
         \centering
         \includegraphics[width=\textwidth]{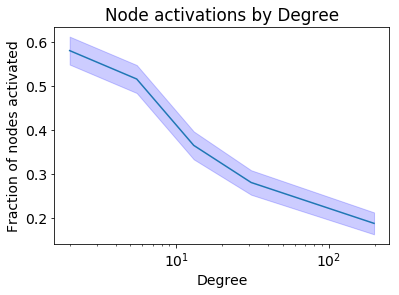}
         \caption{Degree}
         \label{fig:polblogs_degree_activation}
     \end{subfigure}
    \hspace{.05\linewidth}
     \begin{subfigure}[b]{0.33\textwidth}
         \centering
         \includegraphics[width=\textwidth]{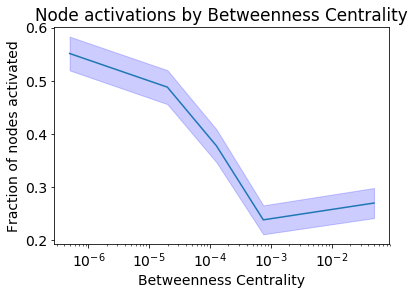}
         \caption{Betweenness Centrality}
         \label{fig:polblogs_centrality_activation}
     \end{subfigure}
    \caption{The above figures show the correlation between individual node features for the \polblogs{} network and prevalence of positive labels, following post-processing at the highest level of $\theta_{min}$ ($\theta_{min} = 0.35$).}
    \label{fig:polblogs_feature_activation_correlation}
\end{figure}
\begin{figure}
     \centering
     \begin{subfigure}[b]{0.32\textwidth}
         \centering
         \includegraphics[width=\textwidth]{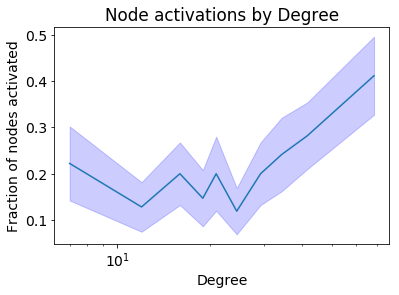}
         \caption{Degree}
         \label{fig:email_degree_activation}
     \end{subfigure}
    \hspace{.05\linewidth}
     \begin{subfigure}[b]{0.33\textwidth}
         \centering
         \includegraphics[width=\textwidth]{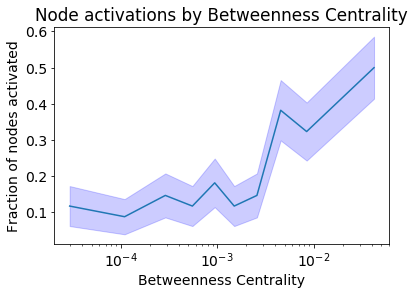}
         \caption{Betweenness Centrality}
         \label{fig:email_centrality_activation}
     \end{subfigure}
    \caption{Analogous to Figure \ref{fig:polblogs_feature_activation_correlation}, the above figure shows the correlation between nodes features and prevalence of positive labels for \emaileu{} ($\theta_{min} = 0.2$)}
    \label{fig:email_feature_activation_correlation}
\end{figure}

\section{Additional details maximizing information access experiments}\label{sec:influence-max-additional}
Assess alternative community detection baselines as well as the impact of the embedding dimension for the information access maximization experiments. 

We assess to the robustness of our results for the maximizing information access experiments in subsection \ref{sec:influence-max} by evaluating the effect of the embedding dimension as well as considering alternative community detection baselines. As discussed in the main text, the community detection baseline that maximizes the exposure of the least exposed group does not have sufficient granularity. As an alternative, we also consider a baseline that maximizes the total welfare of all communities. Let $p_i\left(\mathcal{S} \cup s\right)$ be the probability that node $i$ is informed when an information cascade begins at $\mathcal{S} \cup s$. Then define $\datavec_i\left(\mathcal{S} \cup s\right) = \sqrt{p_i\left(\mathcal{S} \cup s\right)}$. The welfare baseline then greedily maximizes: 
\begin{align}
    \argmax_{s \in V\setminus \mathcal{S}} \sum_{i=1}^n A\left(K, \datavec\left(\mathcal{S} \cup s\right)\right)
\end{align}
Where $K$ is the Louvain community detection kernel. Figure \ref{fig:alterantive-community-detection} compares the performance of the original community detection baseline (``Maximin") against the alternative welfare baseline (``Welfare") and our group-free approach. For each community detection baseline, we also consider an alternative community detection algorithm in greedy modularity maximization. Figure \ref{fig:alterantive-community-detection} shows that the community welfare baseline clearly outperforms the original baseline for \polblogs{} and \lastfm{} whereas the community detection algorithm has little impact. For \polblogs{} and \emaileu{} it is clear that our group-free method achieves the lowest ground-truth inequality. For \lastfm{}, our method achieves inequality on par with the community welfare baseline for large seed sets. 
\begin{figure}
    \centering
    \includegraphics[width=\linewidth]{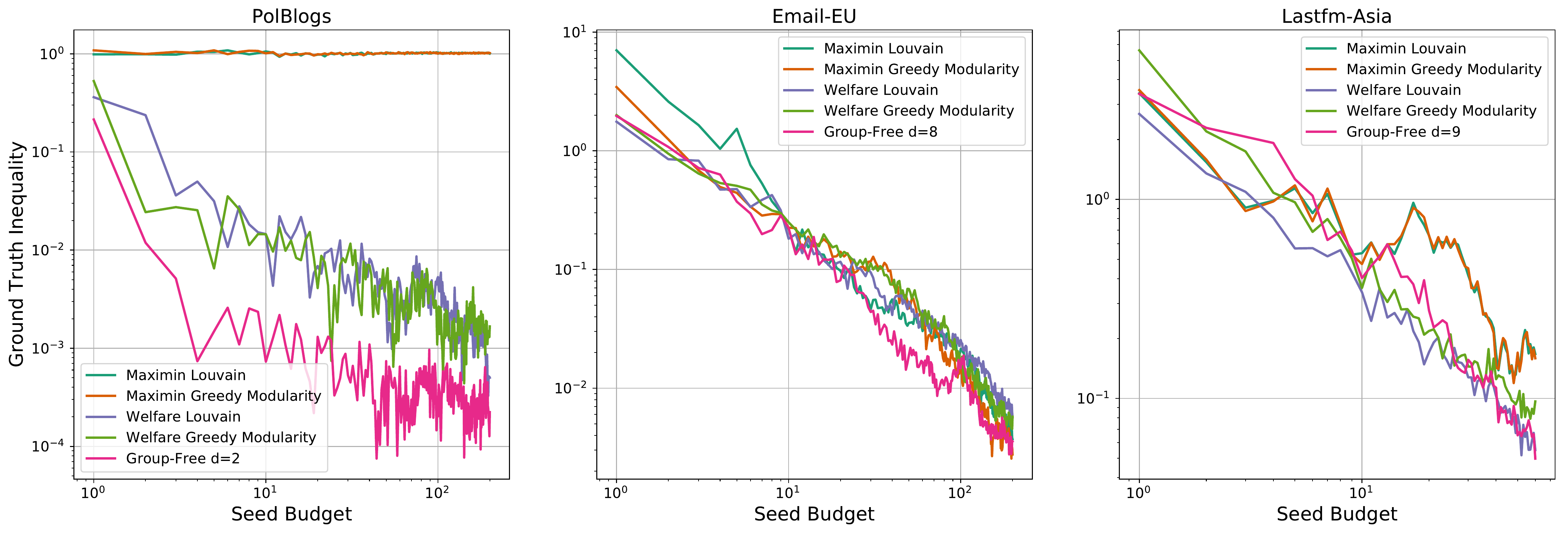}
    \caption{Alternative baselines using community detection}
    \label{fig:alterantive-community-detection}
\end{figure}

We also assess the sensitivity of our results to the choice of the embedding dimension. Figure \ref{fig:alterantive-embedding dimension} shows the performance of the embedding dimension from the main text, which matches the number of ground-truth groups, against alternative dimension values. The Figure confirms that the original number of embedding dimensions performs the best for \polblogs{} and \emaileu{}, and in all cases, kernels constructed with an alternative number of embedding dimensions perform on par.
\begin{figure}
    \centering
    \includegraphics[width=\linewidth]{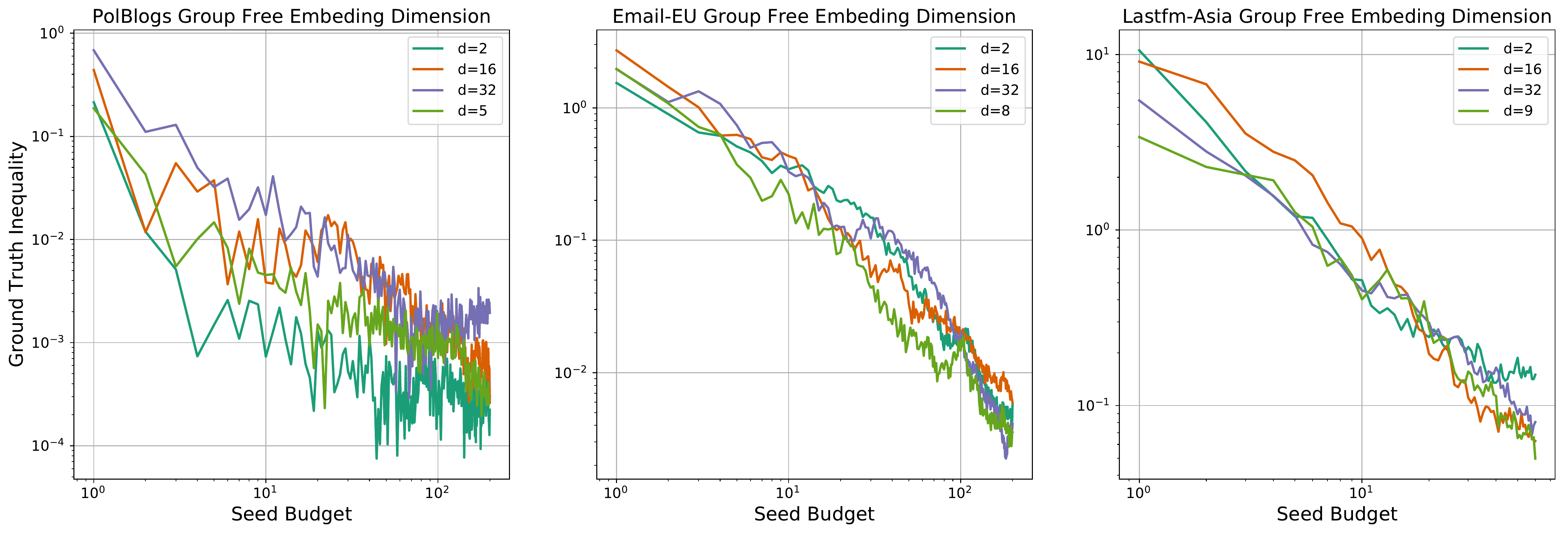}
    \caption{Assess the sensitivity to embedding dimension}
    \label{fig:alterantive-embedding dimension}
\end{figure}

\section{Additional details for fairness in recommendation} \label{app:ranking}

\subsection{Implementation details}

Our experiments are fully implemented in Python 3.9 with PyTorch\footnote{\url{https://pytorch.org/}}. We used the library Networkx\footnote{\url{https://networkx.org/}} to pre-process the networks and for the available implementations of community detection algorithms. 

We preprocessed the Lastfm-Asia dataset with the following steps. First, we extracted all users who listened to at least $50$ artists. Then we removed all user nodes with less than $5$ neighbors, removed groups with less than $100$ users and finally, we retained the largest connected component of users. We obtain a dataset with $n= 2, 785$ users, $m = 7,842$ items, and $17,017$ edges between users. Using country as the group variable, we obtain 9 groups of sizes: 613, 168, 212, 148, 542, 491, 222, 247, 142.

We applied the same preprocessing steps to the Deezer-Europe dataset. We obtain a dataset with $n= 1, 090$ users, $m =  31, 241 $ items, and $ 3, 623 $ edges between users. Using the available binary ``gender'' variable in the dataset, we obtain 2 groups of sizes: 658, 432.

Following \citep{patro2020fairrec,wang2021user,do2021two}, we generate a full user-item preference matrix $(\rho_{ij})_{i,j}$, by applying a standard matrix factorization algorithm\footnote{We chose Alternating Least Squares \citep{hu2008collaborative}, using the implementation of the Python library Implicit (MIT License): \url{https://implicit.readthedocs.io/}} to the incomplete user-item interaction matrix. 

We consider a top-$\nslots$ ranking setting with $\nslots=40$ slots and position weights $b_k = \frac{1}{\log_2(1+k)}$ for $k\leq \nslots$ and $b_k = 0$ for $k \in \{\nslots + 1,\ldots, m\}.$ 

\subsection{Additional baselines and parameters}

We present more results on the Lastfm-Asia dataset. First, we study the sensitivity of the community detection baseline to the choice of algorithm and parameters. In addition to \louvain, we consider the greedy modularity maximization (\greedymodu) algorithm of \cite{clauset2004finding}. For both algorithms, we vary the resolution parameter in $\{0.1,1,10\}$ where values smaller (resp. higher) than 1 favor larger (resp. smaller) communities. The utility-fairness trade-offs achieved are shown in Figure \ref{fig:app-rk}. As claimed in subsection \ref{sec:ranking}, the best results are obtained by \louvain with resolution parameter equal to 1.

Then, we study the sensitivity of our \lapker method to the dimension $d$ of the node embedding space, i.e. the number of eigenvectors of the Laplacian from which the embeddings are constructed. Figure \ref{fig:app-rk} (right) shows that we obtain the best results with $d=17.$

\begin{figure}[t]
    \centering
\includegraphics[width=0.45\linewidth]{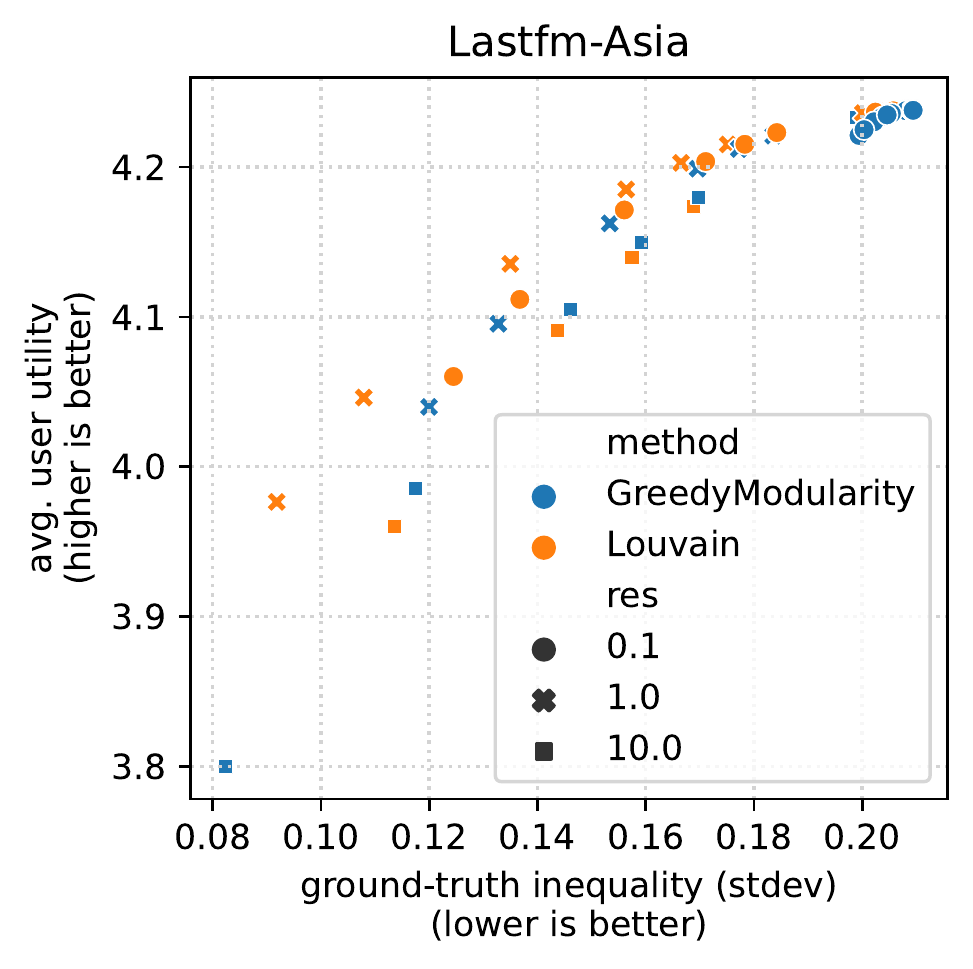}
\includegraphics[width=0.45\linewidth]{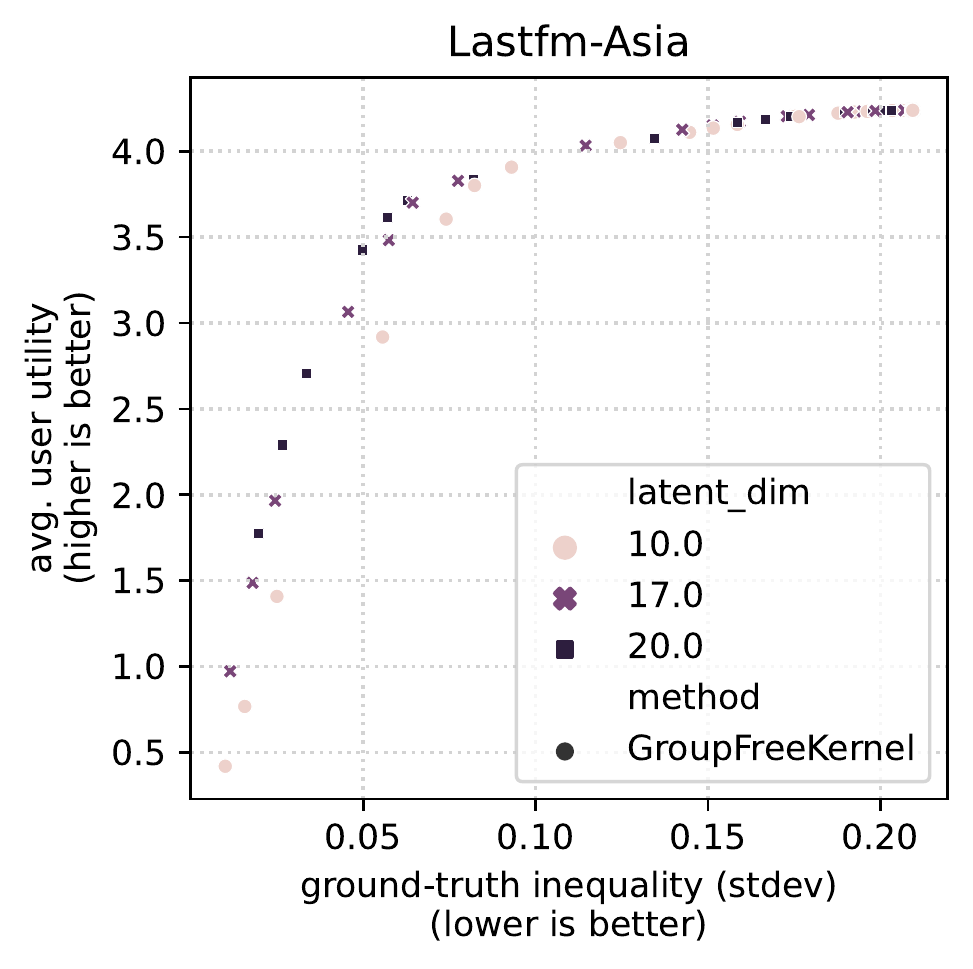}
    \caption{Utility-fairness trade-offs. (Left) Community detection baselines.  (Right) Laplacian kernels.}\label{fig:app-rk}
\end{figure}

\section{Additional related work}\label{sec:additional-related-work}

\paragraph{Fairness without sensitive attributes} Past work has raised concerns regarding the dependence on sensitive attributes in algorithmic fairness definitions. \citet{demographic-risks} details both the individual and group harms of calls to collect additional data in the name of fairness. At the individual level, \citet{demographic-risks} detail risks of privacy violations, misrepresentation or miscategorization, and uses of data beyond informed consent. At the group level, the collection of sensitive data can lead to invasive surveillance, reinforce misleading distinctions between groups, and obscure the variation of attributes over time or within subgroups. A variety of approaches to algorithmic fairness exist to mitigate these concerns. For example, past work has developed methods for fairness that protect the privacy of the sensitive attribute \cite{juarez2022you, jagielski2019differentially, veale2017fairer, kilbertus2018blind}; other methods use noisy protected attributes \cite{mehrotra2021mitigating} or proxies \cite{diana2022multiaccurate, gupta2018proxy}. Our work is the first to leverage homophily in social networks as a correlate for sensitive attributes to derive group-free measures.

\paragraph{Segregation measures}
Geographers and social scientists have developed measures of spatial segregation to quantify, for instance, residential racial segregation in cities \cite{schelling1971} and polarization in voting patterns. The most well-known of these measures is Moran's I. Similar to our work, Moran's I measures segregation without the need for pre-defined groups: given a network, represented as a weight matrix, such as the network of census blocks, and node-level quantities, such as the prevalence of a minority population in a census block, Moran's I quantifies the correlation between node quantities and tie strength, where a higher value indicates that well-connected nodes share similar quantities. The advantage of Moran's I is its ease of computation as the Rayleigh quotient between the weight matrix and the node-quantity vector; however, one limitation is the difficulty of comparing the value of Moran's I across graphs \cite{Duchin2022}. When the weight matrix is the adjacency, for instance, the range of possible values depends on the graph's degree distribution. 

A related measure is assortativity from the network science literature \cite{newman2002assortative, newman2003mixing}. Assortativity is measured as the Pearson correlation coefficient between neighbors for a given node attribute. Though degree is the most common node attribute chosen for assortativity, any node-level quantity is applicable.

\paragraph{Recommender systems} Recently, there has been increasing interest in incorporating fairness criteria in recommendation objectives. Multiple audits have been conducted to assess the fairness of existing recommender systems for their users \citep{sweeney2013discrimination,datta2015automated,imana2021auditing,asplund2020auditing,lambrecht2019algorithmic}. The assessed criteria include showing a job to the same proportion of men and women \citep{imana2021auditing} or showing the same type of housing ads to users belonging to different racial groups \citep{asplund2020auditing}. A few possible formulations of these requirements have been proposed in academic work \citep{usunier2022fast} and a real-world ad system \citep{bogen2023towards} -- our formulation in subsection \ref{sec:ranking} is a group-free version of the ``balanced exposure to user groups'' objective from \citep{usunier2022fast}.

\end{document}